\documentclass[12pt]{amsart}
\usepackage{amssymb,amscd}

\usepackage{tikz-cd}

\usepackage[colorlinks=true,allcolors = blue]{hyperref} 

\let\frak\mathfrak

\def\>{\relax\ifmmode\mskip.666667\thinmuskip\relax\else\kern.111111em\fi}
\def\<{\relax\ifmmode\mskip-.333333\thinmuskip\relax\else\kern-.0555556em\fi}
\def\vsk#1>{\vskip#1\baselineskip}
\def\vv#1>{\vadjust{\vsk#1>}\ignorespaces}
\def\vvn#1>{\vadjust{\nobreak\vsk#1>\nobreak}\ignorespaces}

\let\Medskip\medskip
\def\medskip{\par\Medskip}
\let\Bigskip\bigskip
\def\bigskip{\par\Bigskip}

\let\Maketitle\maketitle
\def\maketitle{\Maketitle\thispagestyle{empty}\let\maketitle\empty}

\newtheorem{thm}{Theorem}[section]
\newtheorem{cor}[thm]{Corollary}
\newtheorem{lem}[thm]{Lemma}

\numberwithin{equation}{section}

\theoremstyle{definition}
\newtheorem*{rem}{Remark}

\let\mc\mathcal
\let\nc\newcommand

\let\al\alpha

\let\la\lambda
\let\La\Lambda

\let\phi\varphi
\let\si\sigma
\let\Si\Sigma

\let\om\omega
\let\Om\Omega

\let\der\partial

\let\ox\otimes

\let\geq\geqslant
\let\le\leqslant
\let\leq\leqslant

\let\on\operatorname
\let\bi\bibitem
\let\bs\boldsymbol

\def\C{{\mathbb C}}
\def\Z{{\mathbb Z}}

\def\B{{\mc B}}
\def\D{{\mc D}}

\def\F{{\mc F}}
\def\L{{\mc L}}
\def\M{{\mc M}}

\def\O{{\mc O}}

\def\V{{\mc V}}
\def\K{{\mc K}}
\def\O{{\mc O}}

\def\+#1{^{\{#1\}}}

\def\End{\on{End}}

\def\const{\on{const}}
\def\ch{\on{ch}}

\def\Wr{\on{Wr}}
\def\Sym{\on{Sym}}

\def\gln{\mathfrak{gl}_N}

\def\beq{\begin{equation}}
\def\eeq{\end{equation}}
\def\be{\begin{equation*}}
\def\ee{\end{equation*}}

\nc{\bea}{\begin{eqnarray*}}
\nc{\eea}{\end{eqnarray*}}
\nc{\bean}{\begin{eqnarray}}
\nc{\eean}{\end{eqnarray}}

\def\h{{\mathfrak h}}

\nc{\Il}{{\mc I_{\bs\la}}}
\nc{\bla}{{\bs\la}}
\nc{\Fla}{\F_{\bs\la}}
\nc{\tfl}{{T^*\Fla}}
\nc{\GL}{{GL_n(\C)}}
\nc{\GLC}{{GL_n(\C)\times\C^*}}

\def\slt{{\frak{sl}_2}}

\nc{\glt}{{\frak{gl}_2}}
\def\KZ/{{\slshape KZ\/}}
\def\qKZ/{{\slshape qKZ\/}}
\def\XXX/{{\slshape XXX\/}}
\nc{\arr}{\rightarrow}
\nc{\larr}{\longrightarrow}
\nc{\A}{{\mc A}}
\nc{\Ax}{{\mc A(\xi)}}
\nc{\cdet}{{\on{cdet}\,}}
\nc{\Fun}{{\on{Fun}_{\slt}\!\!V[0]\,}}
\nc{\Oz}{{\O(\zeta,m,\ell)}}
\nc{\Bm}{{\B(\mu, m,\ell)}}
\nc{\sk}{{\sqrt{-1}}}

\begin{document}

\hrule width0pt
\vsk->
\title[Dynamical Bethe algebra and quasi-polynomials]
{Dynamical $\slt$ Bethe algebra and functions
\\
 on pairs of quasi-polynomials}

\author[A.\,Slinkin, D.\,Thompson, A.\,Varchenko]
{A.\,Slinkin$^{\diamond}$, D.\,Thompson$^*$, A.\,Varchenko$^{\star}$}

\maketitle

\begin{center}
{\it $^{\diamond, *, \star}\<$Department of Mathematics, University
of North Carolina at Chapel Hill\\ Chapel Hill, NC 27599-3250, USA\/}

\vsk.5>
{\it $^{\diamond}\<$
National Research University Higher School of Economics\\ 20 Myasnitskaya Street, 101000 Moscow, Russia\/}

\vsk.5>
{\it $^{\star}\<$Faculty of Mathematics and Mechanics, Lomonosov Moscow State
University\\ Leninskiye Gory 1, 119991 Moscow GSP-1, Russia\/}

\end{center}

\bigskip
\hfill
On the Occasion of the 70th Birthday of Igor Krichiver

\bigskip

\vsk>
{\leftskip3pc \rightskip\leftskip \parindent0pt \Small
{\it Key words\/}: Commuting differential operators, eigenfunctions, Weyl group invariance,
\\
\phantom{aaaaaaaaqq}
Bethe ansatz, Wronskian equation,
quasi-polynomials

\vsk.6>
{\it 2010 Mathematics Subject Classification\/}: 17B80, 81R12, 14M15
\par}

{\let\thefootnote\relax
\footnotetext{\vsk-.8>\noindent
$^\diamond\<$
{\it E\>-mail}: slinalex@live.unc.edu
\\
$^*$
{\it E\>-mail}:
dthomp@email.unc.edu
\\
$^\star\<$
{\it E\>-mail}:
anv@email.unc.edu\,, supported in part by NSF grants
DMS-1665239, DMS-1954266}}

\begin{abstract} 

We consider  the space $\text{Fun}_{\frak{sl}_2}\!V[0]$ of functions
on the Cartan subalgebra of $ \frak{sl}_2$ with values
in the zero weight subspace $V[0]$ of a tensor product of irreducible finite-dimensional
$\frak{sl}_2$-modules. We consider the algebra $\mathcal B$
of commuting differential  operators on $\text{Fun}_{\frak{sl}_2}\,V[0]$, constructed
by  V.\,Rubtsov, A.\,Silantyev, D.\,Talalaev in 2009. We describe the relations between
the action of $\mathcal B$ on $\text{Fun}_{\frak{sl}_2}\,V[0]$ and spaces of pairs 
of quasi-polynomials.

\end{abstract}

{\small\tableofcontents\par}

\setcounter{footnote}{0}
\renewcommand{\thefootnote}{\arabic{footnote}}

\section{Introduction}

A  quantum integrable model is a vector space  $V$
and an algebra $\B$ of commuting linear operators on $V$, 
 called  the Bethe algebra of Hamiltonians. 
The problem is to find eigenvectors and eigenvalues.
If the vector space is a space of functions, then the Hamiltonians are  differential or difference operators.

We say that a quantum integrable model can be geometrized,
 if there is a topological space (a scheme) $X$ with an algebra
$\O_X$ of functions on $X$,  an isomorphism of vector spaces $\psi :\O_X\to V$,  an isomorphism
of algebras $\tau:\O_X\to\B$ such that
\bea
\psi(fg) = \tau(f)\,\psi(g),\qquad \forall f,g\in\O_X.
\eea
These objects $\O_X, \psi,\tau$ identify the  $\B$-module $V$ with the regular representation of the algebra
  $\O_X$ of functions.

If a quantum integrable model $(V,\B)$ is geometrized, then the eigenvectors 
of $\B$ in $V$ are identified with delta-functions of points of $X$ and the eigenvalues of an eigenvector
in $ V$ correspond to evaluations of functions on $X$ at the corresponding  point of $X$.

\smallskip

Our motivation to geometrize the Bethe algebras  came from the examples
considered in \cite{MTV3, MTV5}, where the algebra of Hamiltonians acting on a subspace 
of a tensor product of $\gln$-modules was identified with the
algebra of functions on the intersection of suitable Schubert cycles in a 
Grassmannian. That identification gave an unexpected relation between the representation
theory and Schubert calculus. 

\smallskip

The examples in \cite{MTV3, MTV5} are related to models with  a finite-dimensional vector space $V$.
How to proceed in the infinite-dimensional case of commuting differential operators
 is not clear yet. In this paper we discuss an example. 
 In our infinite-dimensional space $V$
 we distinguish a family of finite-dimensional subspaces $E[\mu]$, $\mu\in\C$, each of which is invariant
 with respect to the algebra $\B$ of commuting differential operators. We geometrize each of the pairs
 $(E[\mu], \B\big\vert_{E[\mu]})$, thus constructing a family of topological spaces $X[\mu]$, $\mu\in\C$. 
 We observe that natural interrelations between the subspaces $E[\mu]$ correspond to natural interrelations
 between the topological spaces $X[\mu]$. For example,  the Weyl involution
$V\to V$, available in our case, identifies $E[\mu]$ and 
 $E[-\mu]$. We show that this identification corresponds to a natural
 isomorphism $X[\mu]\to X[-\mu]$.

\smallskip

Representation theory provides a source of commuting differential or difference operators. In this
paper we discuss the  construction due to 
  V.\,Rubtsov, A.\,Silantyev, D.\,Talalaev, \cite{RST}. That quantum integrable model
   is called the {\it quantum dynamical Gaudin model}. We study the $\slt$ trigonometric version
   of the quantum dynamical Gaudin model, while in \cite{RST}  the $\gln$ elliptic version
was considered. 
\smallskip

Consider the Lie algebra $\slt$ and its Cartan subalgebra $\h\subset \slt$,
$\dim \h = 1$.
For $s=1,\dots,n$, let    $V_{m_s}$ be the irreducible  $\slt$-module of dimension $m_s+1$.
Let  $V=\otimes_{s=1}^nV_{m_s}$,
\bea
V[0] = \{ v\in\ V\ | \ hv=0, \  \ \forall h\in\h\}\,,
\eea
the zero weight subspace.  The space $V[0]$ is nontrivial if $M=\sum_{s=1}^nm_s$ is even.
Let 
$\Fun$  be the
  space of $V[0]$-valued functions on $\h$.
Fix a subset $z=\{z_1,\dots,z_n\}\subset \C^\times$. 
Having these data,  Rubtsov, Silantyev, and Talalaev  construct a family of commuting differential operators acting on $\Fun$.

First,  one  constructs a $2\times 2$-matrix
$\left[
\begin{matrix}
x \der_x & 0  
\\
0 & x \der_x 
\end{matrix} 
\right]+ L(x)\,
=(\delta_{ij}\der_x +  L_{ij}(x))$, 
where  $x$ is a parameter,  
$\der_x =\frac{\der}{\der x}$,
and $L_{ij}(x)$  are differential operators on $\Fun$  
depending on $x$. 
Let
$\mc C = \on{cdet}\big[\delta_{ij}\der_x +  L_{ij}(x)\big],$
where {\it cdet} is the column determinant of the matrix with non-commuting entries,
$\on{cdet}
\left[\begin{array}{cc}
a & b\\
c & d
\end{array}\right] = ad-cb.$
The operator $\mc C$ can be rewritten in the form
\bea
\der_x^2 + C_1(x)\der_x+C_2(x) ,
\eea
where $C_1(x), C_2(x)$  are differential operators on $\Fun$, whose coefficients are rational functions of $x$.
For any $a,b\in \C -\{z_1,\dots,z_n\}$ and $i,j=1,2$ the operators $C_i(a)$, $C_j(b)$ commute. 
The space $\Fun$ with the algebra
$\B$ generated by these commuting differential operators
 is called the {\it quantum dynamical Gaudin model}.

\smallskip
We show that the  algebra $\B$ is generated by the trigonometric KZB operators $H_0$,
$H_1(z)$, \dots, $H_n(z)$, see them in \cite{FW, JV}. 
The KZB operator $H_0$ is also known as the
trigonometric  Hamiltonian operator of the quantum two-particle
Calogero-Moser system with spin space $V$. The operator $H_0$
is the second order differential operator independent of $z$. 
\smallskip

For any $\mu\notin\Z$ we define the subspace $E[\mu]\subset \Fun$ as
the space of meromorphic
eigenfunctions of $H_0$ with eigenvalue $\pi \sqrt{-1}\, \frac{\mu^2}2$ and prescribed poles. The subspaces
$E[\mu]$ were introduced in \cite{FV2} and studied in \cite{JV}. We have $\dim E[\mu]=\dim V[0]$.
The Bethe algebra $\B$ preserves each of $E[\mu]$.

\smallskip
The $\slt$ Weyl involution acts on $\Fun$. The Bethe algebra $\B$ is Weyl group invariant. 
The Weyl involution induces an isomorphism  $E[\mu]\to E[-\mu]$, which is called in \cite{FV2}
the scattering matrix.  The scattering matrix 
$E[\mu]\to E[-\mu]$ is an isomorphism of $\B$-modules.

\smallskip
The basis of the geometrization procedure lies in the following observation. Let $\psi\in E[\mu]$ be an eigenvector of $\B$,
\bea
C_i(x) \psi = E_i(x, \psi) \psi, \qquad i=1,2,
\eea
where $E_i(x,\psi)$ are scalar eigenvalue functions of $x$. We assign to $\psi$ the scalar differential operator
\bea
\mc E_\psi = \der_x^2 + E_1(x,\psi)\der_x+E_2(x,\psi).
\eea
 We show that the kernel of $\psi$ is spanned by 
two quasi-polynomials
$x^{-\mu/2} f(x), x^{\mu/2} g(x)$, where $f(x),g(x)$ are monic polynomials of  degree $M/2$, with the property
that the  Wronskian 
of the two quasi-polynomials  is
\bean
\label{f1}
\Wr(x^{-\mu/2} f(x), x^{\mu/2} g(x))\, = \,\frac {\mu}x\,\prod_{s=1}^n (x-z_s)^{m_s}.
\eean
This fact suggests that the space $X[\mu]$ geometrizing $(E[\mu], \B\big\vert_E[\mu])$
is the space of pairs  $(x^{-\mu/2} f(x)$, $x^{\mu/2} g(x))$ of quasi-polynomials
 with Wronskian given by \eqref{f1}.

 \smallskip
 In this paper we show that this is indeed so. 
We also show that 
 the scattering matrix isomorphism
$E[\mu]\to E[-\mu]$ corresponds to the natural isomorphism $X[\mu]\to X[-\mu]$ defined by the transposition of 
the quasi-polynomials. 

\smallskip
The main message of this paper is the deep relation between the quantum dynamical Gaudin model $(\B, \Fun)$ and
the spaces of pairs of quasi-polynomials.

\smallskip
It would be interesting to develop the elliptic version of this correspondence. In the elliptic version the pairs of quasi-polynomials
are replaced with pairs of theta-polynomials, see \cite{ThV}, but the elliptic KZB operator
$H_0$  does depend on $z$ and does not have apparent analogs of the 
subspaces $E[\mu]$.

\medskip
The paper is organized as follows. In Section \ref{sec 2} we define the
$\slt$ quantum dynamical Gaudin model. 
In Section \ref{sec 3} we discuss properties of the spaces $E[\mu]$. 
In Section \ref{sec 4} we introduce the quantum trigonometric Gaudin model 
$(V[\nu],\B(z,\mu,V[\nu])$ on a weight subspace $V[\nu]\subset V$ and show that
the quantum dynamical Gaudin model $(E[\mu], \B\big\vert_{E[\nu]})$
is isomorphic to the quantum trigonometric  Gaudin model 
$(V[0],\B(z,\mu,V[0])$ on the zero weight subspace. In Section \ref{sec 5} we describe
the Bethe ansatz for the quantum trigonometric  Gaudin model.
In Sections \ref{sec 6} and \ref{sec 7} we describe the kernel of the operator $\mc E_\psi$.
In Sections \ref{sec 8} - \ref{sec 11} we develop the geometrization procedure. The constructions
of Sections \ref{sec 9} - \ref{sec 11} are parallel to the geometrization constructions in
\cite{MTV3, MTV2}.

\medskip

The authors thank V.\,Tarasov for useful discussions.

\section{Quantum dynamical Gaudin model}
\label{sec 2}

\subsection{$\glt$ $RST$-operator}
\label{ssec RST} 
\label{ssec not}

 Consider the complex Lie algebra $\glt$ with standard basis  $e_{11}$, $e_{12}$, $e_{21}$, $e_{22}$.
 Denote by $\h$  the Cartan subalgebra of $\glt$ with basis  $e_{11},e_{22}$ and elements 
 $\la_1 e_{11}+\la_2 e_{22}$.
 Denote 
 \bea
 \la :=\la_1-\la_2.
 \eea
Let $z=\{z_1,\dots,z_n\}\subset \C^\times$ be a set of nonzero pairwise distinct numbers.

\smallskip

Let $V^1,\dots,V^n$ be $\glt$-modules and $V=\otimes_{k=1}^nV^k$.  
Let  $V = \oplus_{\nu\in\h^*} V[\nu]$
be  the weight decomposition, where $ V[\nu] = \{v\in V\ |\ e_{jj}v = \nu (e_{jj})v \ \on{for}\ j=1,2 \}$.
In particular,
\bea
V[0] = \{v\in V\ |\ e_{11}v =  e_{22}v = 0\}.
\eea

For $g\in \glt$, denote $g^{(s)} = 1 \otimes \cdots \otimes g \otimes \cdots \otimes 1\in \End(V)$, with  $g$ in the $s$-th factor.
An element $e_{jk}$ acts on $V$ by $e_{jk}^{(1)}+ \dots + e_{jk}^{(n)}$.
\smallskip

Let $u$ be a variable. Denote
\bea
x= e^{-2\pi \sqrt{-1}u}.
\eea
Let $\der_u = \frac{\der }{\der u}$, \, $\der_x = \frac{\der }{\der x}$,\, $\der_{\la_j} = \frac{\der }{\der \la_j}$ and so on.
\smallskip

Introduce a $2\times 2$-matrix $\L$,
 \bean
\label{matrixL0}
\phantom{aaa}
\\
\notag
  \begin{bmatrix} \L_{11}  &\L_{12}
  \\
   \L_{21}  &\L_{22}
\end{bmatrix} 
 =
 \begin{bmatrix} 
 \pi 
 \sqrt{-1} 
\sum_{s=1}^n \frac{z_s+x}{z_s-x} \,e_{11}^{(s)} 
+ \pi \cot (\pi \la) \, e_{22} \; & \; 
\pi \sqrt{-1} \sum_{s=1}^n \frac{z_s+x}{z_s-x} \,e_{21}^{(s)} - \pi \cot (\pi \la) e_{21} 
 \\
\pi \sqrt{-1} \sum_{s=1}^n \frac{z_s+x}{z_s-x} \,e_{12}^{(s)} + \pi \cot (\pi \la) e_{12} \; 
& \; \pi \sqrt{-1} \sum_{s=1}^n \frac{z_s+x}{z_s-x} \,e_{22}^{(s)} - \pi \cot (\pi \la) e_{11} 
\\
\end{bmatrix},
\eean

\vsk.2>
\noindent
The entries of $\L$ are $\End(V)$-valued trigonometric functions of $u$ and $\la$.

\smallskip
The {\it universal dynamical differential operator} (or the {\it $RST$-operator})
is defined by the formula 
\bean
\label{trigRST}
\mc C = \cdet (\delta_{jk}\partial_u - \delta_{jk} \, \partial_{\la_{j}} + \L_{jk}),
\eean
where for a $2\times 2$-matrix $A = (a_{jk})$ with noncommuting entries the column determinant
 is defined by the formula 
\bea
\cdet A = a_{11}a_{22} - a_{21}a_{12}\,.
\eea
Write the $RST$-operator in the form
\bean
\label{D}
\mc C = \partial_u^2 \,+ \,C_1(x)\, \partial_u \,+ \,C_2(x),
\eean
where $C_1(x)$ and $C_2(x)$ are functions in $x$ with 
values in the space of linear differential operators in variables $\la_1,\la_2$ with coefficients in $\End(V)$.

\begin{thm}
[\cite{RST}]
\label{thm RST}

Fix  $z=\{z_1,\dots,z_n\}\subset \C^\times$.  
Then for any $a \in \C-\{z_1,\dots,z_n\}$ the operators $C_1(a), C_2(a)$, 
restricted to $V[0]$-valued functions of $\la_1,\la_2$,
 define 
linear differential operators in $\la_1, \la_2$ with coefficients in $\End(V[0])$. Moreover, for any 
$a,b \in \C -\{z_1,\dots,z_n\}$, 
the differential operators $C_j(a)$, $C_k(b)$, $j,k=1,2$, acting on the space of $V[0]$-valued functions of $\la_1,\la_2$
commute: 
\bean
\label{cC}
[C_j(a), C_k(b)] = 0, \qquad j,k=1,2.
\eean
\end{thm} 

\smallskip
The  elliptic version of the $RST$-operator for $\gln$ was introduced by V.\,Rubtsov, A.\,Silantyev, D.\,Talalaev in \cite{RST}.
The elliptic version of the $\gln$ $RST$-operator, in particular for the case of $N=2$, 
was discussed in \cite{ThV}. The $RST$-operator,
defined in \eqref{D}, is the trigonometric 
degeneration of the elliptic $\glt$  $RST$-operator.

\subsection{Dynamical Bethe algebra of $\Fun$} 

In this paper, we are interested in  the $\slt$ version of the $RST$-operator.

\smallskip
The Lie algebra $\slt$ is a Lie subalgebra of $\glt$. We have  $\glt=\slt\oplus \C(e_{11}+e_{22})$,
where $e_{11}+e_{22}$ is a central element. 
Let $V^1,\dots,V^n$ be $\slt$-modules, thought of as $\glt$-modules,
 where the central element $e_{11}+e_{22}$ acts by zero.
Let $V=\ox_{k=1}^nV^k$ be the tensor product of the $\slt$-modules. 

\smallskip

 In this paper {\it we consider only such tensor products.}

\smallskip
We consider the Cartan subalgebra of $\slt$ consisting of elements $\la_1e_{11}+\la_2e_{22}$
with $\la_1+\la_2=0$.
 We identify the algebra of functions 
on the Cartan subalgebra of $\slt$ with the algebra of functions in
the variable 
\bea
\la=\la_1-\la_2\,,
\eea
 since the elements
$\la_1e_{11}+\la_2e_{22}$
with $\la_1+\la_2=0$ are uniquely determined by the difference of coordinates.

Denote by $\Fun$ the space of $V[0]$-valued meromorphic functions on the Cartan subalgebra of $\slt$.
In other words, $\on{Fun}_{\slt}\!\!V[0]$ is the space of $V[0]$-valued meromorphic functions in
one variable $\la$.

\smallskip
Each coefficient $C_1(x), C_2(x)$ of the $RST$-operator, defines a differential operator
acting on  $\Fun$. From now on {\it we consider  the coefficients $C_1(x)$, $C_2(x)$ as 
a family of  commuting differential operators  on  $\Fun$}, depending on the parameter $x$.

\smallskip
The commutative algebra of differential operators on  $\Fun$  generated by the identity operator
 and the operators $\{ C_j(a)\ |\ j=1,2, \ a\in \C -\{z_1,\dots,z_n \}\}$
is called the {\it dynamical Bethe algebra} of $\Fun$. The dynamical Bethe algebra depends on the choice of
the numbers $\{z_1,\dots,z_n\}$.

\subsection{Tensor product of $\slt$-modules}
\label{sec TP}

Given $m\in \Z_{\geq 0}$, denote by $V_{m}$ the 
 irreducible $\slt$-module with highest weight $m$. 
 It has a basis $v^m_0,\dots,v^m_m$ such that
 \bean
 \label{V_m}
 \phantom{aaaaa}
 (e_{11}-e_{22})v^m_k=(m-2k)v^m_k,
 \quad
 e_{21}v^m_k=(k+1)v^m_{k+1},
\quad
 e_{12}v^m_k=(m-k+1)v^m_{k-1}.
 \eean
 
 \smallskip
From now on our tensor product $V$ is of the form
\bean
\label{ten V}
V\,=\, \ox_{s=1}^n V_{m_s}\,, \qquad m_s\in\Z_{> 0}\,.
\eean
We have the weight decomposition $V =\oplus_{\nu\in\Z}V[\nu]$ consisting of  weight subspaces
\bean
\label{wdec}
V[\nu] \,=\, \{v\in V\ |\ (e_{11}-e_{22}) v = \nu v\,\}\,.
\eean
If $V[\nu]$ is nonzero, then
\bean
\label{mu m}
\nu = \sum_{s=1}^n m_s - 2k,
\eean
for some nonnegative integer $k$. 
 The dimension of $V[0]$ is positive if the sum $\sum_{s=1}^n m_s$ is even.

 \subsection{Operator $\mc A(\mu)$} 
 \label{sec WG}

The $\slt$ Weyl group $W$ consists of two elements: identity and involution $\si$.
The projective action of $W$ on  $V_m$  is given by the formula
\bea
\si: v^m_k \mapsto  (-1)^kv^m_{m-k}
\eea
for any $k$. We have $\si^2=(-1)^m$. The Weyl group $W$  acts on the tensor product $V$ diagonally.

\smallskip

Following \cite{TV}, introduce
\bea
p(\mu) \,=\, \sum_{k=0}^\infty \,e_{21}^k e_{12}^k\,\frac1{k!}\, 
\prod_{j=0}^{k-1}\,\frac 1{\mu +e_{22}-e_{11}-j}\,,
\qquad \mu\in \C.
\eea 
The series $p(\mu)$ acts on $V_m$, since only a finite number of terms acts nontrivially.
The formula for the action of $p(\mu)$ on a basis vector $v^m_k$
becomes more symmetric if $\mu$ is replaced by $\mu+\frac {\nu}2-1$, where
$\nu=m-2k$ is the weight of $v^m_k$,
\bean
\label{ppro}
p\Big(\mu+\frac {\nu}2-1\Big) v^m_k\,=\,\prod_{j=0}^{k-1}\, \frac{\mu +\frac m2-j}{\mu -\frac m2 + j} \,v^m_k\,,
\eean
see \cite[Section 2.5]{TV}. 

\smallskip
The series $p(\mu)$ acts  on the tensor product $V$ in the standard way. 
Introduce the operator 
\bean
\label{mc A}
\mc  A(\mu) \,:\,V\,\to\,  V,
\quad
v \ \mapsto\ \si p(\mu) \,v\,.
\eean
The  operator $\mc A(\mu)$ is a meromorphic function of $\mu$. For any $\nu$, we have
$\mc A(\mu) V[\nu] \subset V[-\nu]$, and
$\lim_{\mu\to \infty} \mc A(\mu) = \si$\,.  The operator $\mc A(\mu)$ may be considered as a deformation of the Weyl group
operator $\si$.

\begin{lem}
\label{lem v iso}

For  $V= \ox_{s=1}^n V_{m_s}$ as in \eqref{ten V}, denote $M=\sum_{s=1}^n m_s$. 
Assume that $\mu \notin \frac M2+\Z$.
Then for any $\nu$ the  operator
\bean
\label{A mu}
\mc A\Big(\mu +\frac {\nu}2-1\Big)\Big\vert_{V[\nu]}\,:\,V[\nu]\,\to\,V[-\nu]
\eean
is an isomorphism of vector spaces. The composition of the operator
$\mc A\Big(\mu +\frac {\nu}2-1\Big)\Big\vert_{V[\nu]}$ and the operator
\bean
\label{A-mu}
\mc A\Big(-\mu -\frac {\nu}2+1\Big)\Big\vert_{V[-\nu]}\,:\,V[-\nu]\,\to\,V[\nu]
\eean
is the scalar operator on $V[\nu]$ of multiplication by $(-1)^M \frac{\mu -\nu/2}{\mu +\nu/2}$\,.

\end{lem}

\begin{proof} 
The  $\slt$ irreducible decomposition $V=\oplus _m V_m$ of the tensor product $V$
has the highest weights  $m$  of  the form  $m=M-2k$ for $k\in \Z_{\geq 0}$, only.
Now \eqref{A mu} is an isomorphism by formula \eqref{ppro}.
The statement on the composition is \cite[Theorem 10]{TV}.
\end{proof}

\begin{rem}
The operator $\mc A(\mu)$ is the (only) nontrivial element of the $\slt$ dynamical Weyl group of $V$, see definitions in
\cite{EV}.

\end{rem}

\subsection{KZB operators}
Introduce the following elements of $\glt\otimes\glt$\,,
\bea
&&
 \Om_{12} = e_{12}\otimes e_{21},   \phantom{aaaaaaaa}  \qquad \Om_{21} = e_{21}\otimes e_{12}, 
 \\
&&
\Om_0=e_{11}\otimes e_{11} + e_{22}\otimes e_{22},   \qquad
\Om = \Om_0 + \Om_{12} + \Om_{21}. 
\eea
The {\emph{KZB operators}} $H_0, H_1(z), \dots,H_n(z)$ 
are the following differential operators in variables $\la_1,\la_2$ acting on the space $\Fun$\,,
\bean
\label{tKZB} 
&&
\\
\notag
H_0 
&=& 
\frac{1}{4\pi \sqrt{-1}} (\partial_{\la_1}^2 + \partial_{\la_2}^2) + \frac{\pi \sqrt{-1}}{4} 
\sum_{s,t=1}^n \left[ \frac{1}{2} \Om^{(s,t)}_0 + \frac{1}{\sin^2(\pi\la)}
 \left(\Om_{12}^{(s,t)} + \Om_{21}^{(s,t)} \right) \right],
\\
\notag
H_s(z)
&=&
 - (e_{11}^{(s)} \partial_{\la_1} + e_{22}^{(s)} \partial_{\la_2}) + 
\sum_{t:\, t \ne s} \left[ \pi \sqrt{-1} \frac{z_t+ z_s}{z_t-z_s} \Om^{(s,t)} - \pi \cot (\pi \la) 
\left( \Om_{12}^{(s,t)} - \Om_{21}^{(s,t)} \right) \right],
\eean
cf. formulas in Section 3.4 of \cite{JV}.  The elliptic KZB operators were introduced in \cite{FW}.
In \eqref{tKZB} we consider the trigonometric degeneration of the elliptic KZB operators.

\smallskip
By \cite{FW} the operators $H_0, H_1(z), \dots, H_n(z)$ commute and 
$\sum_{s=1}^n H_s(z) = 0$.

\begin{rem}
The differential operator $H_0$ is the {\it Hamiltonian operator of the trigonometric quantum two-particle
system with spin space} $V$.
\end{rem}

\subsection{Coefficients $C_1(x)$, $C_2(x)$}

\begin{lem} 
We have
\bea
C_1(x) = \L^0_{11}(x)+\L^0_{22}(x) - \partial_{\la_1} - \partial_{\la_2}.
\eea
Hence the coefficient $C_1(x)$ acts by zero  on $\Fun$. 
\qed
\end{lem}

\begin{cor} The $RST$-operator \eqref{D} has the form
\bean
\label{mc C}
\mc C=\der_u^2+C_2(x)
\eean
as an operator on $\Fun$.
\end{cor}

\begin{thm} [\cite{ThV}]  
\label{S_2(X)} 

We have 
\bean
\label{C2}
&&
\\
&&
\notag
C_2(x) = -2\pi \sqrt{-1} H_0 - \sum_{s=1}^n 
\left[ 2\pi \sqrt{-1}  \frac{H_s(z)}{1-x/z_s} +  
4\pi^2 \Big(-\frac{c_2^{(s)}}{1-x/z_s} + \frac{c_2^{(s)}}{(1-x/z_s)^2}\Big) \right],
\eean
where $c_2 = e_{11}e_{22} - e_{12}e_{21} + e_{11}$ is a central element of $\glt$.
\end{thm}

\begin{proof}
This is the trigonometric degeneration of the elliptic version of this theorem, see 
  \cite[Theorem 4.9]{ThV}.
\end{proof}

\begin{cor} The dynamical Bethe algebra of $\Fun$ is generated by the identity operator and
the KZB operators $H_0, H_1(z), \dots, H_n(z)$.

\end{cor}

The commutativity of the KZB operators and formulas \eqref{mc C}, 
\eqref{C2} imply the commutativity $[C_2(a),C_2(b)]=0$  independently of  Theorem  
\ref{thm RST}.

\subsection{Weyl group invariance}
\label{sec Wgi}

The Weyl group acts on  $V[0]$
as explained in Section \ref{sec WG}.
Hence the Weyl group acts on $\Fun$  by the formula

\bean
\label{siF}
\si : \psi(\la) \mapsto \si (\psi (-\la)), \qquad
\psi \in\Fun\,.
\eean

\smallskip
\noindent
This action extends to a Weyl group action on
 $\End(\Fun)$, where for  
$T \in \End(\Fun)$ the operator $\si(T)$ is defined as the product 
$\si T\si^{-1}$ of the three elements of $\End(\Fun)$.

\begin{lem} [\cite{ThV}]
\label{weyl inv}
For any  $a\in \C-\{z_1,\dots,z_n\}$ the operator $C_2(a) \in \End(\Fun)$ is Weyl group invariant.
\end{lem}

\begin{proof}
By \cite{FW} the KZB operators  $H_0,H_1(z),\dots,H_n(z)$ are Weyl group invariant.
The lemma  follows from formula \eqref{C2}.
\end{proof}

\section{Eigenfunctions of $H_0$}
\label{sec 3}

\subsection{Trigonometric Gaudin operators}

The trigonometric $r$-matrix is defined by 
\bean
\label{r-matrix}
r(x) = \frac{\Om_+ x + \Om_-}{x-1},
\eean
where
$\Om_+ = \frac12 \Om_0 + \Om_{12}$,\ $\Om_- = \frac12 \Om_0 + \Om_{21}$.

For $\mu \in \C$ the trigonometric Gaudin operators acting on $V$ are defined as 
\bean
\label{trig Gaudin} 
\K_s( z,\mu) = \frac{\mu}{2}\,(e_{11}-e_{22})^{(s)} + \sum_{t:\, t \ne s} r^{(s,t)} (z_s/ z_t), \qquad s=1,\dots,n.
\eean
 Each operator $\K_s(z,\mu)$ preserves each of  the weight subspaces $V[\nu]$
and
 \bea
 [\K_s(z,\mu), \K_t( z,\mu)] = 0
 \eea
  for all $s,t$, see \cite{Ch, EFK}.

\subsection{Dynamical Bethe algebra of  $E(\mu)$} 
\label{sec BE}

Let 
\bea
\La=e^{-2\pi \sqrt{-1} \la}\,, \qquad\on{where}\   \la=\la_1-\la_2\,.
\eea
Let $\A$ be the algebra of functions in $\la$, which can be represented 
as meromorphic functions of $\La$ with poles only at the set $\{\La=1\}$.

\smallskip

For $\mu\in\C$ introduce  the $\A$-module $\A[\mu]$ 
of functions of the form
$e^{\pi \sqrt{-1} \mu\la}f,$ where $f\in\A$.  This module
is preserved by derivatives with respect to $\la_1,\la_2$. Therefore the KZB operator $H_0$ preserves the space
$\A[\mu]\otimes V[0]$.  Any $\psi \in \A[\mu]\otimes V[0]$
has the form
\bea
\psi(\la) = e^{\pi \sqrt{-1} \mu\la}\sum_{k=0}^\infty \La^k \psi^k, \qquad \psi^k\in V[0].
\eea

\begin{thm} [\cite{FV2}]
\label{prop}
Let $\mu \not\in\Z_{>0}$.
Then for any nonzero $v \in V[0]$, there exists a unique $\psi \in \A[\mu] \ox V[0]$ such that 
\bea
H_0\, \psi\, = \,\epsilon \,\psi\,,
\eea
for some $\epsilon\in\C$ and  $\psi^0=v$. Moreover,
$\epsilon = \,\pi \sqrt{-1}\, \frac{\mu^2}2$\,.

\end{thm}

Cf. \cite{JV}.  This function $\psi$ is denoted by $\psi_v$.
\smallskip

We denote by $E[\mu]$ the vector space of functions $\psi\in \A[\mu] \ox V[0]$ such that $H_0\,\psi =\,\pi \sqrt{-1}
\, \frac{\mu^2}2\,\psi$. For more  information on this space see \cite[Section 9]{JV}.

\begin{cor}
 For $\mu \notin \Z_{>0}$, the map 
\bean
\label{is VE}
V[0]\to E[\mu], \qquad v \mapsto \psi_v,
\eean
 is an isomorphism.
\end{cor}

\begin{thm} [\cite{JV}] 
\label{H_s eigen}
Let $\mu \notin \Z_{>0}$. Then for $s=1,\dots,n$, the KZB operators $H_s(z)$ preserve the space $E[\mu]$. 
Moreover, for any $v \in V[0]$ we have
\bea
H_s(z) \psi_v = \psi_w, 
\eea
where $w=-2\pi \sqrt{-1} \, \K_s (z,\mu)\, v$.
\end{thm}

\begin{thm} 
\label{B20 thm}

Let  $\mu \notin \Z_{>0}$, $V= \ox_{s=1}^n V_{m_s}$\,,  and $v\in V[0]$.  Then
\bea
C_2(x)\,\psi_v\,=\,\psi_w\,,
\eea
where
\bean
\label{B20}
\phantom{aaa}
w\,=\, (2\pi \sqrt{-1})^2 \bigg[-\frac{\mu^2}{4} + \sum_{s=1}^n \Big[\frac{m_s(m_s+2)/4+\K_s(z,\mu)}{1-x/ z_s}  
- \; \frac{m_s(m_s+2)/4}{(1-x/z_s)^2} \Big] \bigg] v\,.
\eean
\end{thm}

\begin{proof}
One computes the action of $C_2(X)$ on $\psi_v$ using Theorem \ref{S_2(X)}. The computation is based on 
Theorem \ref{H_s eigen} and the fact that $c_2$ acts on $V_{m_s}$ as multiplication by 
$-\frac{m_s(m_s+2)}{4}$.
\end{proof}

By Theorem \ref{H_s eigen} the subspace $E[\mu]\subset \Fun$ is invariant with respect to the 
action of the dynamical Bethe algebra. The restriction
of the dynamical Bethe algebra  to $E[\mu]$ is called the
{\it dynamical Bethe algebra} of $E[\mu]$ and  denoted by
$\B(z; E[\mu])$.  

\vsk.2>
Notice that $E[\mu]$ 
is a finite-dimensional vector space of dimension $\dim V[0]$. The space $E[\mu]$
does not depend on $ z$,
since the KZB operator $H_0$ does not depend on $ z$. 
The algebra $\B(z; E[\mu])$ is generated by the identity operator and
the KZB operators $H_1(z), \dots, H_s(z)$ and does depend on $ z$.

\subsection{Two-particle scattering matrix}

\begin{thm}
[{\cite[Lemma 6.2]{FV2}}]

\label{thm x=A}

For  $\mu\notin\Z$,  the action \eqref{siF} of the Weyl group involution $\si$ on $\Fun$ identifies
the spaces $E[\mu]$ and $E[-\mu]$. More precisely, for any $v\in V[0]$ we have
\bean
\label{vAw}
\si (\psi_v^\mu(-\la)) \,=\, \psi_{\mc A(\mu-1)v}^{-\mu}(\la),
\eean
where $\psi_v^\mu(\la)$ is the element of $E[\mu]$ with initial term $v$ and 
$\psi_{\mc A(\mu-1)v}^{-\mu}(\la)$ is the element of $E[-\mu]$ with 
initial term $\mc A(\mu-1)v$. Here $\mc A(\mu-1):V[0]\to V[0]$ is the vector isomorphism, defined in \eqref{mc A}.

\end{thm} 

\begin{proof} 
Formula \eqref{vAw} is proved in the example next to Lemma 6.2 in \cite{FV2}.
\end{proof}

\section{Quantum trigonometric Gaudin model}
\label{sec 4}

\subsection{Universal differential operator}
Let $V= \ox_{s=1}^n V_{m_s}$. 
Introduce a $2\times 2$-matrix
\bean
\label{matM}
\mc M = \begin{bmatrix} \mc M_{11} \; & \; \mc M_{12} \\
                              \mc M_{21} \; & \; \mc M_{22} \\
        \end{bmatrix} = -2 \pi \sqrt{-1} \sum_{s=1}^n r^{(0,s)}(x/z_s),
\eean
where $r(x)$ is the trigonometric $r$-matrix defined in \eqref{r-matrix}.
More explicitly,
\bea
\M = \begin{bmatrix} 
2\pi \sqrt{-1} \sum_{s=1}^n \frac{1}{1-x/ z_s} e_{11}^{(s)} - \pi \sqrt{-1} e_{11} \; & 
\; 2 \pi \sqrt{-1} \sum_{s=1}^n \frac{1}{1-x/ z_s} e_{21}^{(s)} - 2 \pi \sqrt{-1} e_{21}  \\
2 \pi \sqrt{-1} \sum_{s=1}^n \frac{1}{1-x/ z_s} e_{12}^{(s)} \; & 
\; 2 \pi \sqrt{-1} \sum_{s=1}^n \frac{1}{1-x/ z_s} e_{22}^{(s)} - \pi \sqrt{-1} e_{22} \\
\end{bmatrix}.
\eea
 The {\it universal (trigonometric) differential operator} for $V$ with parameter $\mu\in\C$
is defined by the formula 
\bean
\label{Dcdet}
\D = \cdet \begin{bmatrix} 
\der_u  -\pi \sqrt{-1} \mu + \M_{11} \; & \; \M_{12} 
\\
      \M_{21} \; & \; \der_u + \pi \sqrt{-1} \mu + \M_{22} 
\\
                    \end{bmatrix}.
\eean
Write the operator $\D$ in the form
\bea
\D = \der_u^2 +D_1(x)\der_u + D_2(x),
\eea
where $D_1(x)$, $D_2(x)$ are $\End(V)$-valued functions of $x$. It is clear that
$\D$ commutes with the action on $V$ of the Cartan subalgebra of $\slt$. In particular,  it means that
$D_1(x)$, $D_2(x)$ preserve the weight decomposition of $V$.

\subsection{Coefficients and Gaudin operators}
\label{sec ctG}

\begin{thm} 
\label{Dthm}
We have  $D_1(x)=0$ and $(2\pi \sqrt{-1})^{-2} D_2(x)$ equals
\bean
\label{DD2}
\phantom{aaaqa}
-\frac{\mu^2+\mu (e_{11} - e_{22}) -  e_{11} e_{22}}4
 + \sum_{s=1}^n \Big[ \frac{m_s(m_s+2)/4 + \K_s ( Z,\mu)}{1-x/z_s} - \frac{m_s(m_s+2)/4}{(1-x/z_s)^2} \Big] .
\eean

\end{thm}

\begin{proof}
The proof is by straightforward calculation. We have
\bea
\D
&=&\Big(\der_u - \pi \sqrt{-1} \mu + 2\pi \sqrt{-1} \sum_{s=1}^n \frac{1}{1-x/ z_s} e_{11}^{(s)} - \pi \sqrt{-1} e_{11}\Big)
 \\
&
\times  
&
\Big(\der_u + \pi \sqrt{-1} \mu + 2 \pi \sqrt{-1} \sum_{s=1}^n \frac{1}{1-x/ z_s} e_{22}^{(s)} - \pi \sqrt{-1} e_{22}\Big) 
\\
&-&
\Big( 2 \pi \sqrt{-1} \sum_{s=1}^n \frac{1}{1-x/ z_s} e_{12}^{(s)}\Big)
 \Big( 2 \pi \sqrt{-1} \sum_{s=1}^n \frac{1}{1-x/ z_s} e_{21}^{(s)} - 2 \pi \sqrt{-1} e_{21} \Big) .
\eea
Then
\bea
D_1(x) =  2 \pi \sqrt{-1} \sum_{s=1}^n \frac{e_{11}^{(s)}+e_{22}^{(s)}}{1-x/ z_s} - \pi \sqrt{-1} (e_{11} + e_{22}) = 0.
\eea
Since  $x=e^{-2 \pi \sqrt{-1} u}$ and $\der_u = -2 \pi \sqrt{-1} \,x \der_x$, the coefficient 
of $(2\pi \sqrt{-1} )^{-2}D_2(x)$ equals
\bean
\label{D/}
&&
 -\frac{\mu^2}{4} - \sum_{s=1}^n \frac{e_{22}^{(s)}}{(1-x/z_s)^2} + \sum_{s=1}^n \frac{e_{22}^{(s)}}{1-x/z_s}
  + \frac{\mu}{2} \sum_{s=1}^n \frac{e_{11}^{(s)}-e_{22}^{(s)}}{1-x/z_s}
   - \frac{\mu}{4}(e_{11}-e_{22}) 
   \\
\notag
&&
+ \sum_{s=1}^n \frac{e_{11}^{(s)} e_{22}^{(s)}}{(1-x/z_s)^2}
 + \sum_{s=1}^n \Big( \sum_{t: \, t\ne s} \frac{e_{11}^{(s)} e_{22}^{(t)} + e_{11}^{(t)} e_{22}^{(s)}}{1- z_s/ z_t} \Big) 
 \frac{1}{1-x/z_s} - \sum_{s=1}^n \frac{e_{11}^{(s)} e_{22}^{(s)}}{1-x/z_s} 
 \\
\notag
 &&
- \sum_{s=1}^n \Big( \sum_{t: \, t\ne s} e_{22}^{(t)} \Big) \frac{e_{11}^{(s)}}{1-x/z_s}
 + \frac14 e_{11} e_{22} - \sum_{s=1}^n \frac{e_{12}^{(s)} e_{21}^{(s)}}{(1-x/z_s)^2} \phantom{aaaaaaaaaa} 
 \\
\notag
 &&
- \sum_{s=1}^n \Big( \sum_{t: \, t\ne s} \frac{e_{12}^{(s)}e_{21}^{(t)} + e_{12}^{(t)} e_{21}^{(s)}}{1- z_s/ z_t} \Big) 
\frac{1}{1-x/z_s} + 
\sum_{s=1}^n \frac{e_{12}^{(s)} e_{21}^{(s)}}{1-x/z_s} + \sum_{s=1}^n \Big( \sum_{t: \, t\ne s} e_{21}^{(t)} \Big) \frac{e_{12}^{(s)}}{1-x/z_s} .
\eean 
The constant term  in \eqref{D/} equals $-\frac{\mu^2+\mu (e_{11} - e_{22}) -  e_{11} e_{22}}4$.
For $s=1,\dots,n$, the coefficient of $\frac{1}{1-x/z_s}$ in \eqref{D/} equals
\bea
&&
-c_2^{(s)}\, +\, \frac{\mu}{2}\, (e_{11}-e_{22})^{(s)} \,+ 
\,\sum_{t: \, t \ne s} \frac{e_{12}^{(s)} e_{21}^{(t)}  z_s 
+ e_{12}^{(t)} e_{21}^{(s)}  z_t}{ z_s -  z_t} + e_{22}^{(s)} \big( e_{22} - e_{22}^{(s)} \big)
\\
&&
\phantom{aaa}
 = \ -c_2^{(s)} \,+ \,\K_s( z,\mu) \,= \,m_s(m_s+2)/4\, +\, K_s(z, \mu).
\eea
The coefficient of $\frac{1}{(1-x/z_s)^2}$ in \eqref{D/} equals 
\bea
-e_{22}^{(s)} + e_{11}^{(s)}e_{22}^{(s)} - e_{12}^{(s)}e_{21}^{(s)} = (e_{11} e_{22} - e_{12}e_{21} + e_{11})^{(s)} = c_2^{(s)}.
\eea
The theorem is proved.
\end{proof}

\begin{lem}
\label{D2 comm}
For any $a,b\in \C-\{z_1,\dots,z_n\}$ the operators
 $D_2(a), D_2(b)\in \End(V)$ commute.
They also commute with  the $\glt$ Cartan subalgebra.

\end{lem}

\begin{proof} 
It is clear that the trigonometric Gaudin operators $\K_s(z,\mu)$ commute with the 
$\glt$ Cartan subalgebra.  Now the lemma follows from the commutativity of
trigonometric Gaudin operators.
\end{proof}

\begin{cor}
\label{cor D2 mu}

Choose a weight subspace $V[\nu]$ of $V$. Then 
$(2\pi \sqrt{-1})^{-2} D_2(x)$ restricted to $V[\nu]$ equals 
\bean
\label{D2}
 -\, \frac{(\mu+\nu/2)^2}{4} + \sum_{s=1}^n \Big[ \frac{m_s(m_s+2)/4 + \K_s ( z,\mu)}{1-x/z_s} - \frac{m_s(m_s+2)/4}{(1-x/z_s)^2} \Big] .
\eean
\qed
\end{cor}

The commutative algebra of  operators on  $V[\nu]$  generated by the identity operator
 and the operators $\{ D_2(a)\ |\ a\in \C -\{z_1,\dots,z_n \}\}$
is called the {\it  Bethe algebra} of $V[\nu]$ with parameter $\mu$ and denoted by
$\B(z;\mu; V[\nu])$.  The  Bethe algebra  $\B(z;\mu; V[\nu])$ 
is generated by the identity operator and the trigonometric Gaudin operators 
$\K_1 ( z,\mu), \dots,\K_n ( z,\mu)$.

\smallskip

The pair  $(V[\nu],  \B(z;\mu; V[\nu]))$ is called the {\it trigonometric Gaudin model on $V[\nu]$}.

\begin{cor}
\label{iso B} 

If $\mu\notin \Z_{>0}$,  the isomorphism $V[0]\to E[\mu]$ in \eqref{is VE} induces an isomorphism
$\B(z;\mu;V[0])\to B(z;E[\mu])$ between the  Bethe algebra of $V[0]$ with parameter $\mu$ and the
dynamical Bethe algebra of the space $E[\mu]$.

\end{cor}

\begin{proof}
The corollary is proved by comparing formulas \eqref{B20} and \eqref{D2}.
\end{proof}

\subsection{Gaudin operators and Weyl group}

\begin{lem}
[{\cite[Lemma 18]{TV}},  cf. {\cite[Lemma 5.5]{MV2}}]
\label{lem AK}
For any weight subspace $V[\nu]$, any $v\in V[\nu]$, $s=1,\dots,n$, we have
\bean
\label{AK}
\mc A\Big(\mu +\frac {\nu}2-1\Big) \K_s(z,\mu) v\,=\,
 \K_s(z, - \mu)\mc A\Big(\mu +\frac {\nu}2-1\Big) v.
 \eean

\end{lem}

\begin{thm}
\label{thm isom mu}

For  $V= \ox_{s=1}^n V_{m_s}$ as in \eqref{ten V}, denote $M=\sum_{s=1}^n m_s$. 
Assume that $\mu \notin \frac M2+\Z$. 
Then for any $\nu$ the isomorphism of vector spaces
\bean
\label{A mu n}
\mc A\Big(\mu +\frac {\nu}2-1\Big)\Big\vert_{V[\nu]}\,:\,V[\nu]\,\to\,V[-\nu]
\eean
induces  an isomorphism of  Bethe algebras
\bean
\label{B iso}
\phantom{aaa}
\B(z;\mu; V[\nu]) \to  \B(z;-\mu; V[-\nu]), \quad
T\mapsto\mc A\Big(\mu +\frac {\nu}2-1\Big) T
\mc A\Big(\mu +\frac {\nu}2-1\Big)^{-1}.
\eean

\end{thm}

\begin{proof}
The theorem is a corollary of Lemmas    \ref{lem v iso} and \ref{lem AK}.
\end{proof}

\subsection{Commutative diagram}

Assume that $\mu\notin\Z$ and $M$ is even. Then $V[0]$ is a nonzero weight subspace.

\smallskip

Consider the $\B(z;\mu;V[0])$-module $V[0]$ and $\B(z;-\mu;V[0])$-module $V[0]$.
 Consider the $\B(z; E[\mu])$-module $E[\mu]$ and $\B(z; E[-\mu])$-module $E[-\mu]$.
Consider the diagram relating these modules
\bean
\label{comD}
\begin{tikzcd}
(\B(z;\mu;V[0]),\, V[0])   \arrow[r, ] \arrow[d, ] & (\B(z;-\mu;V[0]), \,V[0]) \arrow[d, ] 
\\
(\B(z; E[\mu]),\, E[\mu]) \arrow[r,   ]  & (\B(z; E[-\mu]),\, E[-\mu])
\end{tikzcd}\ \  .
\eean

\noindent
Here the map
$(\B(z;\mu;V[0]),\, V[0]) \to (\B(z;-\mu;V[0]), \,V[0])$ is the module isomorphism of Theorem \ref{thm isom mu}.
The map 
$(\B(z; E[\mu]),\, E[\mu]) \to (\B(z; E[-\mu]),\, E[-\mu])$ is the module isomorphism induced by the action
of the Weyl involution $\si$
and the fact that the $RST$-operator is Weyl group invariant, see Lemma \ref{weyl inv}.
The maps 
$(\B(z;\mu;V[0]),\, V[0])\to (\B(z; E[\mu]),\, E[\mu])$ and  
$(\B(z;-\mu;V[0]),\, V[0])\to (\B(z; E[-\mu]),\, E[-\mu])$
are the module isomorphisms of Corollary \ref{iso B}.

\begin{thm}
\label{thm tra}
Diagram \eqref{comD} is commutative.

\end{thm}

\begin{proof}
The theorem follows from Theorems \ref{H_s eigen}, \ref{thm x=A}, \ref{thm isom mu}.
\end{proof}

\section{Bethe ansatz}
\label{sec 5}

\subsection{Bethe ansatz equations for  triple $(z;\mu;V[\nu])$}
\label{sec BAE}

Let
$V= \ox_{s=1}^n V_{m_s}$, as in \eqref{ten V},
and $M=\sum_{s=1}^n m_s$.  Let $V[\nu]$ be a nonzero 
weight subspace of $V$.
Then $\nu = M-2m$ for some nonnegative integer $m$.
\smallskip

Let $z=\{z_1,\dots,z_n\}\subset \C^\times$ be a set of nonzero pairwise distinct numbers,
as in Section \ref{ssec not}. Let $\mu\in\C$.
\smallskip

Introduce the  {\emph{master function}} of the variables $t=(t_1,\dots,t_m), \mu, z$,
\bea
\Phi(t,z,\mu) \, 
&=&
  \, \Big(1-\mu + \frac{\nu}{2} \Big) \sum_{i=1}^m \,\ln t_i \,+ 
  \,\sum_{s=1}^n \,\frac{m_s}{4}\, (2\mu + m_s - \nu) \,\ln z_s     
 \\
& + &
2 \sum_{1 \leqslant i < j \leqslant m} \ln(t_i - t_j) 
- \sum_{i=1}^m  \sum_{s=1}^n m_s \ln(t_i - z_s) + 
\sum_{1 \le s < r \leqslant n} \frac{m_s m_r}{2} \ln(z_s - z_r).
\eea
The {\emph{Bethe ansatz equations}} are the critical point equations for
the master function $\Phi(t,z,\mu)$ with respect to the variables $t_1,\dots,t_m$,
\bean
\label{tr.BAE}
\frac{1-\mu+\nu/2}{t_i} \,+ \,\sum_{j \ne i} \frac{2}{t_i-t_j}\, -\, \sum_{s=1}^n \frac{m_s}{t_i-z_s}\, = \,0, 
\qquad i=1,\dots,m.
\eean
The master function $\Phi(t,z,\mu)$ is the trigonometric degeneration 
 of the elliptic master function considered in Section 5 of \cite{ThV}, see also \cite{FV1, MaV}.
 
\smallskip
The symmetric group $S_m$ acts on the critical set. If $(t_1^0,\dots, t^0_m; z;\mu)$
 is a solution of the Bethe ansatz equations, then for any $\rho\in S_m$, 
 the point $(t_{\rho(1)}^0,\dots, t^0_{\rho(m)};z;\mu)$ is also a solution.

\subsection{Bethe vectors}
Define
\bea
\mc{C} &=& \{\ell = (\ell_1,\dots,\ell_n)\in\Z_{\geq 0}\ |\  \ell_s\leq m_s, \, \ell_1+\dots+\ell_n=m\},
\\
\om_{\ell} (t,z) 
&=&
 \Sym \; \Big[
\prod_{s=1}^n  \prod_{i=\ell_1+\dots+\ell_{s-1}+1}^{\ell_1+\dots+\ell_s} \frac{1}{t_i-z_s}\Big],
\eea
where $\Sym f(t_1,\dots,t_m) = \sum_{\rho \in S_m} f(t_{\rho(1)}, \dots,t_{\rho(m)})$.
Introduce the {\it weight function} 
\bean
\label{wght_f}
\om(t,z) \,=\, \sum_{\ell \in \mc{C}} \,\om_{\ell} (t,z)\, 
v^{m_1}_{\ell_1}\otimes \dots\otimes v^{m_n}_{\ell_n}\,,
\eean
see Section \ref{sec TP}.  This weight function see in \cite{MV2}, also in \cite{JV, MaV, SV}.
\smallskip

Notice that the weight function is a symmetric function of the variables $t_1,\dots,t_m$.

\smallskip
 If $(t^0\!\,;z;\mu)$ is a solution of the Bethe ansatz equations \eqref{tr.BAE}, 
 then the vector $\om(t^0\!, z)$ is called the {\it Bethe vector}.

\begin{thm}
[\cite{MTV6,V}]
\label{thm Bnon}

Let $(t^0\!; z; \mu)$ be a solution of the Bethe ansatz equations \eqref{tr.BAE}.
Then the Bethe vector $\om(t^0\!, z)$ is nonzero.

\end{thm}

\begin{thm} [\cite{FV1, JV}, cf. \cite{RV}] 
\label{eigenv}

Let $(t^0 ; z; \mu)$ be a solution of the Bethe ansatz equations \eqref{tr.BAE}.
 Then the Bethe vector $\om(t^0\!, z)$ is an eigenvector of the trigonometric Gaudin operators,
\bea
\K_s(z,\mu) \, \om(t^0\!,z) \,=\, z_s\,\frac{\der \Phi}{\der z_s}(t^0\!, z,\mu) \, \om(t^0\!, z)\,,
\qquad s=1,\dots,n.
\eea
\end{thm}

Denote 
\bean
\label{k_s}
\phantom{aaa}
k_s(t^0\!,z,\mu)
&=&
 z_s\,\frac{\der \Phi}{\der z_s}(t^0\!,z,\mu)
 \\
 \notag
 &=&
  \frac{m_s}{2} \Big[ (\mu - \nu/2 + m_s/2)\, +\, 
 \sum_{p: \,p \ne s} m_p \,\frac{z_s}{z_s-z_p}
\, + \,
2 \sum_{i=1}^m \frac{z_s}{t_i^0-z_s} \Big]. 
 \eean

\subsection{Bethe vectors and coefficient $D_2(x)$}

\begin{lem}
\label{hatB2 lem}

If $(t^0\!\,; z; \mu)$ is a solution of the Bethe ansatz equations 
\eqref{tr.BAE}, then the Bethe vector $\om(t^0\!,z)$ is an eigenvector of
all operators of the Bethe algebra $\B(z;\mu; V[\nu])$. In particular, the 
operator $D_2(x)$ acts on $\om(t^0\!,z)$ by multiplication by the scalar
\bean
\label{sc D2}
\phantom{aaaa}
(2\pi \sqrt{-1})^2\Big[- \frac{(\mu+\nu/2)^2}{4} + \sum_{s=1}^n \Big[ \frac{m_s(m_s+2)/4 + k_s (t^0\!, z,\mu)}{1-x/z_s}
 - \frac{m_s(m_s+2)/4}{(1-x/z_s)^2} \Big]\Big] .
\eean

\end{lem}

\begin{proof}
The lemma follows from Theorem \ref{eigenv} and Corollary \ref{cor D2 mu}.
\end{proof}

For  a  solution  $(t^0\!; z; \mu)$  of the Bethe ansatz equations 
\eqref{tr.BAE}, we introduce the {\it fundamental differential operator}
\bean
\label{hatD}
\mc E_{\,t^0\!,z,\mu} = \der_u^2 + E_2(x, t^0\!,z,\mu),
\eean
where the function $E_2(x, t^0\!,z,\mu)$ is given by formula \eqref{sc D2}.

\subsection{Basis of Bethe vectors}

The Bethe ansatz method is the method to construct eigenvectors of commuting operators,
see Lemma  \ref{hatB2 lem} as an example. The standard problem 
is to determine if  the Bethe ansatz method gives a basis of eigenvectors  of the vector
space, on which the commuting operators act.
In the case of Lemma  \ref{hatB2 lem} the answer is positive.

\begin{lem} 
\label{basis}

Let $\mu\notin \frac\nu 2+\Z_{>0}$.  Then for generic $z=\{z_1,\dots,z_n\}\subset \C^\times$,
the set of solutions $(t^0;z;\mu)$ 
 of   system \eqref{tr.BAE}
of the Bethe ansatz equations
 is such that
 the corresponding Bethe vectors $\om(t^0\!,z,\mu)$ form a basis of the space $V[\nu]$.

\end{lem}

\begin{proof}
Here the word generic means 
that the subset of all acceptable  sets $\{z_1,\dots,z_n\}$ forms a Zariski open subset in the space
of all sets $\{z_1,\dots,z_n\}$. The proof of the lemma is standard. It is a modification of \cite[Theorem 8]{ScV},
cf. \cite[Section 4.4]{MV1}, \cite[Section 5.4]{MV2}, \cite[Section 10.6]{MTV1}.
\end{proof}

\section{Function $w(x)$ in the kernel of  $\mc E_{\,t^0;z;\mu}$}
\label{sec 6}

Let $(t^0;z;\mu)$  be a solution of system \eqref{tr.BAE} of  Bethe ansatz equations, where
$t^0=(t^0_1,\dots,t^0_m)$. Define 
\bean
\label{Y&u}
y(x) \,=\, \prod_{i=1}^m (x- t^0_i), \qquad  
w(x) \,=\,y(x)\, x^{\frac{\nu/2-\mu}{2}} \,\prod_{s=1}^n (x-z_s)^{-m_s/2}\,.
\eean

\begin{thm}
\label{Dthm}

We have
\bean
\label{Dform}
\mc E_{\,t^0;z;\mu}\, =\, \big( \der_u + (\ln w)' \big) \big( \der_u - (\ln w)' \big).
\eean
where  $' = \der/\der u$. In other words,  
\bean
\label{hatB2form}
E_2(x, t^0\!,z,\mu)  = - (\ln w)'' - ((\ln w)')^2.
\eean
\end{thm}

\begin{rem}
For $\nu=0$ this statement is the trigonometric degeneration of its elliptic version \cite[Theorem 5.3]{ThV}.
\end{rem}

\begin{proof} Recall that $\partial_u\, =\, -2\pi \sqrt{-1} \,x \partial_x$. We have 
\bea
(\ln w)' 
&=&
 -2\pi \sqrt{-1} \,\Big[-  \frac{\nu/2+\mu }{2} + \sum_{i=1}^m \frac{t^0_i}{x-t^0_i}
  - \frac12 \sum_{s=1}^n \frac{z_s m_s}{x- z_s} \Big], 
\\
(\ln w)'' 
&=&
 (2 \pi \sqrt{-1} )^2 \Big[ - \sum_{i=1}^m \frac{t^0_i}{x-t^0_i} 
 - \sum_{i=1}^m \frac{(t^0_i)^2}{(x-t^0_i)^2} 
 + \frac12 \sum_{s=1}^n \frac{z_s m_s}{x-z_s} 
 + \frac12 \sum_{s=1}^n \frac{z_s^2 m_s}{(x-z_s)^2} \Big].
\eea
Hence, $(2\pi \sqrt{-1})^{-2}(- (\ln w)'' - ((\ln w)')^2)$  equals
\bea
&&  \sum_{i=1}^m \frac{t^0_i}{x-t^0_i} 
 + \sum_{i=1}^m \frac{(t^0_i)^2}{(x-t^0_i)^2} 
 - \frac12 \sum_{s=1}^n \frac{z_s m_s}{x-z_s} 
 - \frac12 \sum_{s=1}^n \frac{z_s^2 m_s}{(x-z_s)^2} 
 \\
 && 
 \phantom{aa}
- \frac14 (\mu+\nu/2)^2 - \sum_{i=1}^m \frac{(t_i^0)^2}{(x-t^0_i)^2} 
- 2 \sum_{i=1}^m  \sum_{j:\,j \ne i} \frac{t^0_it^0_j}{t^0_i - t^0_j}  \frac{1}{x-t^0_i} 
\\
&&
 \phantom{aa}
- \frac14 \sum_{s=1}^n \frac{z_s^2 m_s^2}{(x-z_s)^2}
- \frac 12\sum_{s=1}^n   \sum_{p:\,p \ne s} \frac{z_s z_pm_s m_p}{z_s-z_p} 
 \frac{1}{x-z_s} + (\mu+\nu/2) \sum_{i=1}^m \frac{t^0_i}{x-t^0_i} 
\\
&&
 \phantom{aa}
-\frac12\,(\mu+\nu/2)\sum_{s=1}^n \frac{z_s m_s}{x-z_s}
+ \sum_{i=1}^m \sum_{s=1}^n \frac{t^0_i z_s m_s}{t^0_i-z_s}  \frac{1}{x-t^0_i} -
\sum_{i=1}^m  \sum_{s=1}^n \frac{t^0_i z_s m_s}{t^0_i-z_s}  \frac{1}{x-z_s} .
\eea
In the expression above for each $i=1,\dots,m$ the coefficient of $\frac{1}{(x-t^0_i)^2}$ 
equals zero. The coefficient of $\frac{1}{x-t^0_i}$ equals
\bea
&&
(\mu+\nu/2+1)t^0_i - 
\sum_{j:\,j \ne i} \frac{2t^0_it^0_j}{t^0_i-t^0_j} 
+ \sum_{s=1}^n \frac{t^0_i z_sm_s}{t^0_i-z_s} 
\\
&&
 \phantom{aa}
= t^0_i 
\Big[ \mu+\nu/2+1  +  2 \sum_{j:\,j \ne i} \frac{t^0_j-t^0_i+t^0_i}{t^0_j-t^0_i}
 - \sum_{s=1}^n \frac{(z_s - t^0_i+t^0_i)m_s}{z_s-t^0_i} \Big] 
 \\
&&
 \phantom{aa}
= t^0_i \Big[ \mu+\nu/2+1  + 2(m-1) + \sum_{j:\,j \ne i} \frac{2t^0_i}
{t^0_j-t^0_i} - \sum_{s=1}^n m_s - \sum_{s=1}^n \frac{t^0_i m_s}{z_s-t^0_i} \Big]
\\
&&
 \phantom{aa}
= -(t^0_i)^2 \Big[ \frac{1-\mu+\nu/2}{t^0_i} + 
\sum_{j:\,j \ne i} \frac{2}{t^0_i-t^0_j} - 
\sum_{s=1}^n \frac{m_s}{t^0_i-z_s} \Big] = 0, 
\eea
where the last equality follows from
the Bethe ansatz equations \eqref{tr.BAE}. For each $s=1,\dots,n$
 the coefficient of $\frac{1}{(1-x/z_s)^2}$ equals
  $-m_s(m_s+2)/4$. The coefficient of $\frac{1}{1-x/z_s}$ equals
\bea
&&
\frac12 m_s + \frac12 m_s \sum_{p:\,p \ne s} \frac{z_p m_p}{z_s-z_p} 
+ \frac12 (\mu+\nu/2) m_s + \sum_{i=1}^m \frac{t^0_i m_s}{t^0_i-z_s} 
\\
&&
 \phantom{aa}
= \frac{m_s}{2} \Big[ 1+ \mu+ \nu/2 - \sum_{p:\,p \ne s} 
\frac{(z_p-z_s+z_s)m_p}{z_p-z_s} + 2 \sum_{i=1}^m \frac{t^0_i-z_s+z_s}{t^0_i-z_s} \Big] 
\\
&&
 \phantom{aa}
= \frac{m_s}{2} \Big[ 1+ \mu+ \nu/2 - \sum_{p \ne s} m_p - \sum_{p:\,p \ne s} 
\frac{z_s m_p}{z_p-z_s} +
 2\sum_{i=1}^m1 + 2 \sum_{i=1}^m \frac{z_s}{t^0_i-z_s} \Big] 
\\
&&
 \phantom{aa}
= \frac{m_s}{2} \Big[ (1+ \mu - \nu/2 + m_s) + \sum_{p:\,p \ne s}
 \frac{z_s m_p}{z_s-z_p} + 2 \sum_{i=1}^m \frac{z_s}{t^0_i-z_s} \Big] 
\\
&&
 \phantom{aa}
= m_s(m_s+2)/4 + \frac{m_s}{2} \Big[ (\mu - \nu/2 + m_s/2) + \sum_{p:\,p \ne s}
 m_p \frac{z_s}{z_s-z_p} + 2 \sum_{i=1}^m \frac{z_s}{t^0_i-z_s} \Big] 
 \\
&&
 \phantom{aa}
= m_s(m_s+2)/4 + k_s(t^0,z,\mu), 
\eea
where $k_s(t^0,z,\mu)$ are defined in \eqref{k_s}. Hence, $ E_2 = - (\ln w)'' - ((\ln w)')^2$.
\end{proof}

\begin{cor}
The function $w(x)$  lies in the kernel of $\mc E_{\,t^0;z;\mu}$.
\end{cor}

\section{Function $\tilde w(x)$ in the kernel of $\mc E_{\,t^0;z;\mu}$}
\label{sec 7}

\subsection{Wronskian}

 \label{ssec Wr}
The {\emph{Wronskian}} of two  functions $f(a)$ and $g(a)$ is 
\bean
\label{Wr}
\Wr_a (f,g) = f \frac{d g}{d a} - \frac{d f}{d a} g.
\eean
We have 
\bean
\label{h^2}
\Wr_a(hf, hg) = h^2 \Wr(f,g)
\eean
for any function $h(a)$.

\subsection{Wronskian  and Bethe ansatz equations}

\begin{lem} 
\label{lem WBA}

The following two statements hold:

\begin{enumerate}

\item[(i)]

Let $\mu\notin \frac\nu2+\Z_{\geq 0}$. Let   $(t^0;z;\mu)$ be a solution
 of the Bethe ansatz equations \eqref{tr.BAE} and
$y(x)=  \prod_{i=1}^m(x-t^0_i)$.  Then there exists a unique 
monic polynomial $\tilde y(x)$ of degree $M-m$, such that 
\bean
\label{Wr.eqn22}
\Wr_x (y(x),x^{\mu-\nu/2} \tilde y(x)) \,=\,\const\, x^{\mu-\nu/2-1} \prod_{s=1}^n (x-z_s)^{m_s},
\eean
where $\const$ is a nonzero constant.

\item[(ii)]

Let  $\mu\ne \frac\nu2$. Assume that $y(x)=  \prod_{i=1}^m(x-t^0_i)$ is a polynomial
with distinct roots
 such that   $y(z_s)\ne 0$,  $s=1,\dots,n$. Assume that
 there exists a polynomial $\tilde y(x)$ such that
 equation \eqref{Wr.eqn22} holds. Then $(t^0_1,\dots,t^0_m;z;\mu)$ is a solution of the Bethe ansatz equations \eqref{tr.BAE}.
\end{enumerate}

\end{lem}

\begin{proof}
This lemma is a reformulation of Theorem 3.2 and Corollary 3.3 in \cite{MV2}.
\end{proof}

\subsection{Function $\tilde w(x)$} \label{ssec utild}

Recall that we have  a solution $(t^0; z; \mu)$ of the Bethe ansatz equations, the differential
operator $\mc E_{\,t^0\!,z,\mu}$  and the function
$w(x) \,=\,y(x)\, x^{\frac{\nu/2-\mu}{2}} \,\prod_{s=1}^n (x-z_s)^{-m_s/2}$, where
$y(x) =\prod_{i=1}^m (x- t^0_i)$.

\begin{thm}
\label{quasipthm}

Let $\mu\notin \frac\nu2+\Z_{\geq 0}$. Then there exists a unique 
monic polynomial $\tilde y(x)$ of degree $M-m$, such that the function
\bean
\label{utild}
\tilde w(x) \,= \,\tilde y(x)\, x^{\frac{\mu-\nu/2}{2} }\, \prod_{s=1}^n (x-z_s)^{-m_s/2}
\eean 
lies in the kernel of $\mc E_{\,t^0; z;\mu}$. The functions $w(x), \tilde w(x)$ span the kernel of
 $\mc E_{\,t^0 ; z; \mu}$.

\end{thm}

\begin{proof}

The differential operator $\mc E_{\,t^0; z; \mu}$ introduced in \eqref{hatD} has no first order term. 
Hence the kernel of $\mc E_{\,t^0; z; \mu}$ consists of  the functions $\tilde w(x)$ satisfying  the equation
\bean
\label{WSo}
\Wr_u (w(x), \tilde w(x)) = \const\,.
\eean
 By Lemma \ref{lem WBA},  there exists a unique monic polynomial $\tilde y(x)$ of degree $M-m$, such that 
equation \eqref{Wr.eqn22} holds.
Dividing both sides of \eqref{Wr.eqn22} by $x^{\mu-\nu/2} \prod_{s=1}^n (x-z_s)^{m_s}$ we obtain
\bean
\label{Wr.eqn2}
\phantom{aaa}
\Wr_x\! \Big( y(x)\, x^{\frac{\nu/2-\mu}{2}} \prod_{s=1}^n (x-z_s)^{-m_s/2},\,
\tilde y(x) \,x^{\frac{\mu-\nu/2}{2} } \prod_{s=1}^n (x-z_s)^{-m_s/2} \Big) = \,\const\,x^{-1}.
\eean
Recall that $x=e^{-2\pi \sqrt{-1} u}$, hence $\der_u  = -2\pi \sqrt{-1}\, x \der_x$. This implies equation \eqref{WSo}.
The theorem is proved.
\end{proof}

\subsection{Bethe ansatz equations for triples $(z;\mu; V[\nu])$ and $(z; -\mu; V[-\nu])$}

\begin{lem}
\label{lem yty}
Let $\mu\notin \frac\nu2+\Z_{\geq 0}$.
Let $(t^0; z; \mu)$ be a solution of the Bethe ansatz equations \eqref{tr.BAE} assigned to
the triple $(z; \mu; V[\nu])$ in Section \ref{sec BAE}.  Let 
\bean
\label{roty}
\tilde y(x) \,=\, \prod_{i=1}^{M-m}(x-\tilde t^{\,0}_i)
\eean
be the polynomials assigned to $(t^0; z; \mu)$ in Theorem \ref{quasipthm}. If
 $\tilde y(x)$ has distinct roots and $\tilde y(z_s)\ne 0$ for $s=1,\dots,n$,
 then
$(\tilde t^{\,0}_1,\dots,\tilde t^{\,0}_{M-m}; z; -\mu)$ is a solution of the Bethe ansatz equations 
\eqref{tr.BAE} assigned to the triple $(z; -\mu;  V[-\nu])$.

\end{lem}

\begin{proof}
Equation \eqref{Wr.eqn22} can be rewritten as 
\bean
\label{Wr2}
\Wr_x (x^{-\mu+\nu/2}y(x), \tilde y(x)) \,=\,\const\, x^{-\mu+\nu/2-1} \prod_{s=1}^n (x-z_s)^{m_s}.
\eean
Now the lemma follows from the equalities $-\nu = -M+2m=M-2(M-m)$
and Lemma \ref{lem WBA}.
\end{proof}

\begin{thm}
[{\cite[Theorem 5.7]{MV2}}]
\label{thm WB}

 Under assumptions of Lemma \ref{lem yty} consider 
the Bethe vectors  $\om(t^0\,,z, \mu) \in V[\nu]$ and
 $\om(\tilde t^{\,0}\!,z,-\mu) \in V[\nu]$.  Then
 \bean
 \label{BVA}
\mc A\Big(\mu +\frac {\nu}2-1\Big)
\om(t^0\,,z,\mu) =\, \const\,
\om(\tilde t^{\,0}\!,z,-\mu),
\eean
where $\const$ is a nonzero constant.

\end{thm}

\begin{cor}
\label{cor =eig}
 Under assumptions of Lemma \ref{lem yty}, for $s=1,\dots,n$,
 the eigenvalue of $\K_s(z,\mu)$ on $\om(t^0\!,z,\mu)$ equals 
 the eigenvalue of $\K_s(z,-\mu)$ on $\om(\tilde t^{\,0}\!,z,\mu)$.
\end{cor}

\begin{proof}
The corollary follows from Lemma \ref{lem AK} and Theorem \ref{thm isom mu}.
\end{proof}








\section{Conjugates of $\D$ and  $\mc E_{\,t^0; z; \mu}$}
\label{sec 8}

\subsection{Conjugate of $\D$}

Recall the universal differential operator $\D=\der_u^2+D_2(x)$  introduced in \eqref{Dcdet}, where
the coefficient $D_2(x)$ is determined by formula \eqref{DD2}. 
We introduce the operator 
\bean
\label{Dconj}
\D^{c} = \frac{1}{(2\pi \sqrt{-1}\,x)^2}\,\prod_{s=1}^n (x-z_s)^{m_s/2} \cdot \D \cdot \prod_{s=1}^n (x-z_s)^{-m_s/2},
\eean
where the superscript $^c$ stays for the word ``conjugated''.

\begin{thm}
\label{Dconjthm}
We have 
\bean
\label{Dconjform}
\phantom{aaaaaa}
\D^{c} 
&=&
 \der_x^2 + \Big[ \frac{1}{x} - \sum_{s=1}^n \frac{m_s}{x- z_s} \Big] \der_x     
-  \frac{1}{x} \sum_{s=1}^n \frac{m_s/2}{x- z_s} + \sum_{s=1}^n \frac{m_s(m_s+2)/4}{(x- z_s)^2}
\\ 
 &+&
    \sum_{s \ne p} \frac{m_s m_p/4}{(x- z_s)(x- z_p)}  
-\frac{\mu^2+\mu (e_{11} - e_{22}) -  e_{11} e_{22}}{4x^2}    
 \notag
 \\  
&-&
\frac{1}{x^2} \,\sum_{s=1}^n \left[  z_s \frac{m_s(m_s+2)/4 + 
\K_s(z,\mu)}{x- z_s} +  z_s^2 \frac{m_s(m_s+2)/4}{(x- z_s)^2} 
\right].
 \nonumber
\eean

\end{thm}

\begin{proof} 

 Recall that $x = e^{-2 \pi \sqrt{-1} u}$, 
 $\partial_u = -2 \pi \sqrt{-1} \,x \partial_x, \; \partial_u^2 = (2\pi \sqrt{-1})^2 ( x \partial_x + x^2 \partial_x^2 )$. 
 Denote \\
 $f = \prod_{s=1}^n (x-z_s)^{-m_s/2}$.  We have
$f' =   -\sum_{s=1}^n \frac{m_s/2}{x-z_s}  \,f\,$,
\bea
f'' =  \Big( \sum_{s=1}^n \frac{m_s^2/4}{(x-z_s)^2} +
 \sum_{s \ne p} \frac{m_s m_p / 4}{(x-z_s)(x-z_p)} + \sum_{s=1}^n \frac{m_s/2}{(x-z_s)^2} \Big) f,
\eea
where $' = \der/\der x$. Therefore,
\bea
\D^{c}
& =& \frac{1}{x^2} f^{-1} \Big[x^2 \der_x^2 + x\der_x + \frac{1}{(2\pi \sqrt{-1})^2} D_2(x) \Big] f 
\\
&=&
 \frac{1}{x^2} f^{-1} \Big[ x^2 \big( f \der_x^2 + 2 f' \der_x + f'' \big) 
 + x \big(f \der_x + f' \big) + \frac{1}{(2\pi\sqrt{-1} )^2} D_2(x) f \Big]
\\
&=&
 \der_x^2 + \Big[ 2 f^{-1} f' + \frac{1}{x} \Big] \der_x + \Big[ f^{-1} f'' + \frac{1}{x} f^{-1} f' + 
 \frac{1}{x^2} \frac{1}{(2\pi \sqrt{-1})^2} D_2(x) \Big],
\eea
which gives the right-hand side of formula \eqref{Dconjform}.
\end{proof}

\subsection{Conjugate of $\mc E_{\,t^0;z;\mu}$\,} Similarly to the conjugation of $\D$ we conjugate 
$\mc E_{\,t^0\!,z,\mu}$ and consider the differential operator
\bean
\label{conj Dt}
\mc E_{\,t^0; z; \mu}^c = \frac{1}{(2\pi \sqrt{-1}\,x)^2}\,\prod_{s=1}^n (x-z_s)^{m_s/2} \cdot 
\mc E_{\,t^0; z; \mu} \cdot \prod_{s=1}^n (x-z_s)^{-m_s/2}.
\eean

\begin{lem}
\label{kerDconj}

The kernel of $\mc E_{\,t^0; z; \mu}^c$ is spanned by quasi-polynomials 
\bean
\label{qp PQ}
 x^{\frac{\nu/2-\mu}{2}} y(x),
 \qquad
  x^{\frac{\mu-\nu/2}{2} }\tilde y(x),
  \eean
where $y(x)$ is the monic polynomial of degree $m$, defined in \eqref{Y&u}, and $\tilde y(x)$ is the monic polynomial of degree
$M-m$, defined in Theorem \ref{quasipthm}.
\qed  
\end{lem} 

\begin{lem}
\label{lem e=e}
Under assumptions of Lemma \ref{lem yty}, let
$(t^0; z; \mu)$ be a solution of the Bethe ansatz equations \eqref{tr.BAE} assigned to
the triple $(z; \mu; V[\nu])$. Assume that 
the numbers $\tilde t^{\,0}_1,\dots,\tilde t^{\,0}_{M-m}$ defined in Lemma \ref{lem yty}
are such that 
$(\tilde t^{\,0}_1,\dots,\tilde t^{\,0}_{M-m}; z; -\mu)$ is a solution of the Bethe ansatz equations 
\eqref{tr.BAE} assigned to the triple $(z; -\mu;  V[-\nu])$. Then
\bean
\label{e=e}
\mc E_{\,\tilde t^{\,0}; z; -\mu}^c
=
\mc E_{\,t^0; z; \mu}^c\,.
\eean
\qed
\end{lem}

\section{Space of $V$-valued functions of $z_1,\dots,z_n$}
\label{sec 9}

\subsection{Space $V_1^{\ox n}[\nu]$}
\label{sec V1}

Recall the two-dimensional irreducible $\slt$-module $V_1$ 
with basis 
$v^1_0\,,\,v^1_1$\,,  see \eqref{V_m}.
In the rest of the paper we assume that $V$ is the tensor power of $V_1$, 
\bean
\label{V1n}
V\,=\, V_1^{\ox n}\,, \qquad\on{where} \ \ n>1.
\eean
The space $V$ has a basis of vectors
\bea
v_I=v_{i_1}^1\ox\dots\ox v_{i_n}^1\,,
\eea
labeled by partitions $I=(I_1,I_2)$ of $\{1,\dots, n\}$, where
 $\,i_j =0$ if $j\in I_1$, and $\,i_j =1$ if $j\in I_2$.
 We have the weight decomposition $V=\oplus_{m=0}^n V[n-2m]$, where
$V[n-2m]$
is of dimension $\binom{n}{m}$ and
has the basis  $\{v_I\ |\ I=(I_1,I_2), \ |I_1|=m, \ |I_2|=n-m\}$.

\smallskip
We use notations \
$\nu=n-2m$, \ $\ell = n-m$,\ and hence\ $ m+\ell=n$.

\subsection{Space $\V^{S}$} 
Let $z=(z_1,\dots,z_n)$ be variables. The symmetric group 
$S_n$  acts on the algebra $\C[z_1,\dots,z_n]$ by permuting the variables. Let $\si_s(z)$,
$s=1,\dots,n$, be the $s$-th elementary symmetric polynomial in $z_1,\dots, z_n$.
 The algebra of
symmetric polynomials $\C[z_1,\dots, z_n]^S$ is a free polynomial algebra with generators 
$\si_1(z),\dots,\si_n(z)$.

\smallskip

Let $\V$ be the space of polynomials in $z_1,\dots,z_n$ with coefficients
in $V_1^{\ox n}$,
\bea
\V = V_1^{\ox n}\ox \C[z_1,\dots,z_n].
\eea
The symmetric group $S_n$ acts on $\V$ by permuting the factors of $V_1^{\ox n}$
and the variables $z_1,\dots,z_n$ simultaneously,
\bea
\rho(v_1\ox\dots\ox v_n\ox p(z_1,\dots,z_n))=
v_{(\rho^{-1})(1)}\ox\dots\ox v_{(\rho^{-1})(n)}\ox p(z_{\rho(1)}, \dots, z_{\rho(n)}),\quad \rho\in S_n.
\eea
We denote by $\V^S$ the subspace of $S_n$-invariants in $\V$.

\begin{lem}
[\cite{MTV3}]
 The space $\V^S$ is a free $\C[z_1,\dots, z_n]^S$-module of rank $2^n$.
\end{lem}

Consider the grading on $\C[z_1,\dots,z_n]$
 such that $\deg z_s = 1$ for all $s = 1,\dots, n$. 
 We define a
grading on $\V$ by setting $\deg(v \ox p) = \deg p$ for any $v \in V_1^{\ox n}$ and 
$p\in\C[z_1,\dots,z_n]$. The
grading on $\V$ induces a  grading on $\End(\V)$.

\smallskip
The Lie algebras $\slt\subset\glt$ naturally act on $\V^S$. We have the weight decomposition
\bea
\V^S=\oplus_{m=0}^n \V^S[n-2m],
\qquad
\V^S[n-2m] = (V[n-2m]\ox\C[z_1,\dots,z_n])^{S}\,.
\eea

\smallskip
Let $M$ be a $\Z_{>0}$-graded space with finite-dimensional homogeneous components. Let
$M_j\subset M$ be the homogeneous component of degree $j$. The formal power series in a
variable $\al$,
$\ch_M(\al) =\sum_{j=0}^\infty (\dim M_j)\, \al^j,$\,
is called the graded character of $M$.

\begin{lem}
[\cite{MTV2}]
\label{lem frV}
The space $\V^S[n-2m]$  is a free $\C[z_1,\dots, z_n]^S$-module of rank $\binom{n}{m}$
and
\bean
\label{ch V}
\ch_{\V^S[n-2m]}(\al) \,=\, \prod_{i=1}^m \frac 1{1-\al^i} \cdot \prod_{i=1}^{n-m} \frac 1{1-\al^i}\,.
\eean
\end{lem}

\subsection{Bethe algebra of $\V^S[\nu]$}

Recall the differential operator $\D^c$ introduced in \eqref{Dconjform}
for $V=\oplus_{s=1}^nV_{m_s}$ and depending on  parameter $\mu\in \C$.
 For $V=V_1^{\ox n}$ the operator $\D^c$ takes the form
\bean
\label{Dz}
\mc F= \der_x^2 + F_1(x)\der_x + F_2(x),
\eean
where
\bean
\label{B}
\phantom{aaa}
F_1(x) 
&=&
  \frac{1}{x} - \sum_{s=1}^n \frac{1}{x- z_s}\, ,
\\
\notag
F_2(x) 
&=&
-  \frac{1}{x} \, \sum_{s=1}^n \frac{1/2}{x- z_s} + \sum_{s=1}^n \frac{3/4}{(x- z_s)^2}
+ \sum_{s \ne p} \frac{1/4}{(x- z_s)(x- z_p)}
\\ 
 &-&
    \frac{\mu^2+\mu (e_{11} - e_{22}) -  e_{11} e_{22}}{4x^2}    
-
\frac{1}{x^2} \,\sum_{s=1}^n \left[  z_s \frac{3/4 + 
\K_s(z,\mu)}{x- z_s} +  z_s^2 \frac{3/4}{(x- z_s)^2} 
\right].
\notag
\eean
In formula \eqref{Dconjform} we had $\{z_1,\dots,z_n\}$ being
 a subset of $ \C^\times$. From now on we consider
$z_1,\dots,z_n$ as independent variables.

\smallskip
The operator $\mc F$ in formula \eqref{Dz} with variables $z_1,\dots,z_n$ 
is called the {\it universal differential operator} for $\V^S$ with parameter $\mu\in\C$.

\begin{lem}
[{cf. {\cite[Section 2.7]{MTV3}}}]
\label{lem LG}
The Laurent expansions of $F_1(x)$ and $F_2(x)$ at infinity have the form
\bean
\label{Fij}
F_1(x) = \sum_{j=1}^\infty F_{1j} x^{-j}\,, 
\qquad  
F_2 (x) = \sum_{j=2}^\infty F_{2j} x^{-j}\,,
\eean
where 
\bea
F_{11} 
&=&
1 - n, 
\qquad 
F_{1j} = - \sum_{s=1}^n z_s^{j-1} \quad \on{for} \; j\geq 2 .  
\eea
For any $j\geq 2$, the element $F_{2j}$ is a homogeneous polynomial
in $z_1,\dots,z_n$ of degree $j-2$ with coefficients in
$\End(V)$. The element $F_{2j}$ preserves the weight decomposition of
$\V$. Each of the elements  $F_{1j}$, $j\geq 1$, $F_{2j}$, $j\geq 2$, defines an endomorphism 
of the $\C[z_1,\dots,z_n]^{S}$-module $\V^S$.
\end{lem}

\begin{proof}
The proof follows from  straightforward calculations.
\end{proof}

\begin{lem}
The elements $F_{1j}$, $j\geq 1$, $F_{2j}$, $j\geq 2$, considered as endomorphisms of
 the $\C[z_1,\dots,z_n]^{S}$-module  $\V^S$, commute.
\end{lem}

\begin{proof}
The commutativity follows from the commutativity of the trigonometric Gaudin operators
in formula \eqref{Dconjform}.
\end{proof}

For a  weight subspace $\V^S[\nu]$, $\nu = n-2m$, $\ell = n-m$,
consider the commutative subalgebra
$\B(\mu; m;\ell)$ 
of the algebra of endomorphisms of the  $\C[z_1,\dots,z_n]^{S}$-module  $\V^S[\nu]$,
generated by the elements $F_{1j}$, $j\geq 1$, $F_{2j}$, $j\geq 2$.
The subalgebra $\B(\mu; m;\ell)$ 
 is called the {\it Bethe algebra} of $\V^S[\nu]$ with parameter $\mu\in \C$.

\begin{lem}
\label{lem z in B}
The Bethe algebra $\B(\mu; m;\ell)$ contains the subalgebra of
operators of  multiplication by
elements of $\C[z_1,\dots,z_n]^S$.
\end{lem}

\begin{proof}
The subalgebra of  operators of multiplication by elements of 
$\C[z_1,\dots,z_n]^S$ is generated by
the elements $F_{1j}$, $j\geq 1$, see Lemma \ref{lem LG}.
\end{proof}

Lemma \ref{lem z in B} makes the Bethe algebra $\B(\mu; m;\ell)$
a $\C[z_1,\dots,z_n]^S$-module.

\subsection{Weyl group invariance}

For a weight subspace $V[\nu]=V_1^{\ox n}[\nu]$ recall   the operator 
$\mc A\big(\mu+ \nu/2-1\big) : V[\nu] \to V[-\nu]$,  defined in \eqref{mc A}.
It is an isomorphism of vector spaces,  if $\mu \notin \frac n2 + \Z$.
That operator induces an isomorphism of  $\C[z_1,\dots,z_n]^S$-modules,

\bean
\label{mc Amu}
\mc A\big(\mu+ \nu/2-1\big) : \V^S[\nu] \to \V^S[-\nu].
\eean
\vsk.2>

\begin{lem}
\label{lem BB iso}
Let $\mu \notin \frac n2 + \Z$. Let $F_{ij}(\mu, m,\ell)$ be the generators
of $\B(\mu;m;\ell)$, defined in \eqref{Fij},  and
$F_{ij}(-\mu,\ell, m)$  the generators
of $\B(-\mu;\ell;m)$.  Then
\bean
\label{Fij iso}
F_{ij}(-\mu,\ell, m) \,=\,
\mc A\Big(\mu +\frac {\nu}2-1\Big) F_{ij}(\mu, m,\ell)\,
\mc A\Big(\mu +\frac {\nu}2-1\Big)^{-1}
\eean
for all $i,j$. 
The map
 
\bean
\label{mu-mu}
\B(\mu;m;\ell) \to \B(\mu;m;\ell), \quad F_{ij}(\mu; m; \ell)
\mapsto
F_{ij}(-\mu,\ell, m),
\eean

\smallskip
\noindent
is an isomorphism of algebras and of $\C[z_1,\dots,z_n]^S$-modules.
The maps  in \eqref{mc Amu} and \eqref{mu-mu} define an isomorphism
between the $\B(\mu;m;\ell)$-module $\V^S[\nu]$
and the $\B(-\mu;\ell;m)$-module $\V^S[-\nu]$.

\end{lem}

\begin{proof}
The lemma follows from Lemma \ref{lem AK}.
\end{proof}

\subsection{Generic fibers of $\V^S[\nu]$}

Given $a=(a_1,\dots,a_n)\in \C^n$, denote by $I_a\subset \C[z_1,\dots, z_n]$ the ideal
generated by the polynomials  $\si_s(z)-a_s$, $s=1,\dots,n$. Define
\bean
\label{ide}
 I_a\V^S[\nu] \,:=\, \V^S \cap (V[\nu]\ox I_a).
\eean
Assume that $a$ is such that the polynomial $x^n + \sum_{s=1}^n(-1)^s a_sx^{n-s}$ has distinct
nonzero roots  $b_1,\dots,b_n$.

\begin{lem}
[{\cite[Lemma 2.13]{MTV3}}]
\label{lem fib}

The quotient $\V^S[\nu]/ I_a\V^S[\nu]$ is a finite-dimensional complex vector space canonically isomorphic
to $V[\nu]$. Under this isomorphism
the  Bethe algebra $\B(\mu; m;\ell)$ induces a commutative algebra of  operators on 
$V[\nu]$. That commutative algebra of operators is canonically isomorphic to the Bethe algebra
$\B(b_1,\dots,b_n; \mu; V[\nu])$ introduced in Section \ref{sec ctG}.

\end{lem}

\section{Functions on pairs of quasi-polynomials}
\label{sec 10}

\subsection{Space of pairs of quasi-polynomials}

Let $m,\ell,n$ be positive integers, $m+\ell=n$. Denote
$\nu=n-2m$, cf. Section \ref{sec V1}.
Let 
\bea
\zeta\,\in \,\C - \frac 12\,\Z\,.
\eea
 Let $\Om(\zeta,m,\ell)$ be the affine $n$-dimensional space
with coordinates $p_i$, $i=1,\dots,m$,  $q_j$, $j=1,\dots,\ell$.
Introduce the generating functions
\bean
\label{pq}
p(x) &=& x^{-\zeta}\,(x^m + p_1x^{m-1} + \dots + p_m),
\\
\notag
q(x) &=& x^{\zeta}\,(x^\ell + q_1x^{\ell-1} + \dots + q_\ell).
\eean
We identify points $U$ of $ \Om(\zeta, m,\ell)$ with  two-dimensional complex vector spaces 
generated by quasi-polynomials
\bean
\label{pqU}
p(x,U) &=& x^{-\zeta}\,(x^m + p_1(U)x^{m-1} + \dots + p_m(U)),
\\
\notag
q(x,U) &=& x^{\zeta}\,(x^\ell + q_1(U)x^{\ell-1} + \dots + q_\ell(U)).
\eean
Denote by $\O(\zeta,m,\ell)$
the algebra of regular functions on $\Om(\zeta,m,\ell)$,
\bea
\O(\zeta,m,\ell) = \C[p_1,\dots,p_m, q_1,\dots,q_\ell].
\eea
 Define the  grading on 
$\O(\zeta,m,\ell)$ by $\deg p_i=\deg q_i=i$ for
all $i$.

\begin{lem}
\label{lem grO}
The graded character of the algebra $\O(\zeta, m,\ell)$ equals
\bean
\ch_{\O(\zeta, m,\ell)} (\al) = \prod_{i=1}^m \frac{1}{1-\al^i} \cdot \prod_{j=1}^\ell \frac{1}{1-\al^j}.
\eean
\qed
\end{lem}

\subsection{Wronski map}

Let $p(x), q(x)$  be the generating functions in \eqref{pq}.  We have
\bean
\label{Wpq}
\Wr_x(p,q) = \frac {2\zeta + \ell-m}x\,
\Big(x^n + \sum_{s=1}^n\,(-1)^s\,\Si_s \,x^{n-s}\Big),
\eean
where $\Si_1,\dots,\Si_n$ are elements of $\O(\zeta, m,\ell)$. 
 Notice that 
$2\zeta + \ell-m = 2\zeta +\nu\,\notin\Z$ according to our assumptions. 
The elements  $\Si_1,\dots,\Si_n$ are homogeneous with
$\deg \Si_s = s$.

\smallskip
Define the {\it Wronski map}
\bea
\Wr\, :\, \Om(\zeta, m,\ell) \to \C^n, \quad
U \mapsto 
(\Si_1(U), \dots, \Si_n(U)).
\eea

\begin{lem}
\label{lem pdeg}
For  $\zeta\in \C-\frac 12\Z$, \,the Wronski map is a map of positive degree. 

\end{lem}

\begin{proof} 
The proof is a slight modification of the proof of \cite[Proposition 3.1]{MTV4}. 
\end{proof}

Let $\O^S \subset \Oz$ be the subalgebra generated by $\Si_1,\dots,\Si_n$.
Let $\si_1,\dots,\si_n$ be coordinates on  $\C^n$, which is the image of the Wronski
map. Introduce the grading on $\C[\si_1,\dots,\si_n]$ by $\deg \si_s=s$ for all $s$.
The Wronski map induces the isomorphism
$\C[\si_1,\dots,\si_n]\to \O^S$, $\si_s \mapsto \Si_s$, of graded algebras, see Lemma
\ref{lem pdeg}.
This isomorphism makes $\Oz$ a $\C[\si_1,\dots,\si_n]$-module.

\subsection{Another realization of $\O(\zeta, m,\ell)$}

Define the differential operator $\mc G$ by
\bean
\label{DO1}
\mc G = \frac{1}{\Wr_x(p,q)} \, 
\on{rdet}\begin{bmatrix} p & p' & p'' 
\\
 q & q' & q'' 
 \\ 
 1 & \der_x & \der^2_x \end{bmatrix},
\eean
where $\on{rdet}$ is the row determinant.
 We have
\bean
\label{DO2}
\mc G = \der_x^2 + G_1(x) \der_x + G_2(x),
\eean

\medskip
\noindent
cf. \cite{MTV3}.
It is a differential operator in variable 
$x$ and \  $ G_1(x)$, $G_2(x)$  are rational functions in $x$  with coefficients in $\O(\zeta, m,\ell)$.

\smallskip
Notice  that
\bean
\label{G1}
G_1 \,=\,-\,\frac{(\Wr_x(p,q))'}{\Wr_x(p,q)}\,.
\eean 
\vsk.2>

\begin{lem}
[{cf. {\cite[Section 2.7]{MTV3}}}]
\label{lem LGG}
The Laurent expansions of $G_1(x)$ and $G_2(x)$ at infinity have the form
\bean
\label{LaG} 
G_i (x) = \sum_{j=i}^\infty G_{ij} x^{-j}, \qquad  i=1,2, 
\eean
where for
any $i,j$, the element $G_{ij}$ is a homogeneous element of $\O(\zeta, m,\ell)$
 of degree $j-i$.
\end{lem}

\begin{proof}
The proof is by  straightforward calculation.
\end{proof}




\begin{lem} 
[{\cite[Lemma 3.4]{MTV3}}, {\cite[Lemma 4.3]{MTV2}}]

\label{O_thm}
Let  $\zeta\in \C-\frac 12\Z$.
Then  the elements $G_{ij}$,   $i=1,2$,  $j\geq i$,
 generate the algebra $\O(\zeta, m,\ell)$. 
 \qed
 \end{lem}
 
 \subsection{Fibers of  Wronski map}

Given $a=(a_1,\dots,a_n)\in \C^n$, denote by $J_a\subset \O(\zeta,m,\ell)$ the ideal
generated by the elements  $\Si_s-a_s$, $s=1,\dots,n$. Define
\bean
\label{ide}
 \O_a(\zeta,m,\ell) \,:=\, \O(\zeta,m,\ell)\big/ J_a .
 \eean

\noindent
 The algebra $\O_a(\zeta,m,\ell)$ is the algebra of functions on the fiber $\Wr^{-1}(a)$ of the Wronski map.

\smallskip
Let 
\bean
\label{ab}
x^n+\sum_{s=1}^n\,(-1)^{n-s}\,a_s\,x^{n-s} =   \prod_{s=1}^n(x-b_s)
\eean
 for some $b_s\in\C$.  Let  $U=\langle p(x,U), q(x,U)\rangle$ be a point of
$\Om(\zeta,m,\ell)$ and
\bea
p(x,U)=x^{-\zeta}\prod_{i=1}^m(x-t^0_i),
\qquad
q(x,U)=x^{\zeta}\prod_{i=1}^\ell(x-\tilde t^{\,0}_i),
\eea
for some $t_i^0, \tilde t^{\,0}_i\in\C$.

\begin{lem}
\label{lem gen f}
Let $\zeta\in \C-\frac 12\Z$. Then there exists
 a Zariski open subset $X\subset \C^n$ such that for any  $a\in X$ 
 all the numbers $b_1,\dots,b_n$ are nonzero and distinct.
Moreover, for any point $U\in\Wr^{-1}(a)$ all the numbers
$b_1,\dots,b_n$, $t^0_1,\dots,t^0_m$, $\tilde t^{\,0}_1,\dots,\tilde t^{\,0}_\ell$ are distinct.
\qed

\end{lem}

\begin{lem}
\label{cor gen f}
If $a\in X$ and $U\in \Wr^{-1}(a)$, then
$(t^0_1,\dots,t^0_m; b_1,\dots,b_n; 2\zeta + \nu/2)$ is a solution
of the Bethe ansatz equations \eqref{tr.BAE} assigned to the triple
$(b_1,\dots,b_n; 2\zeta + \nu/2; V[\nu])$, and
$(\tilde t^{\,0}_1,\dots,\tilde t^{\,0}_\ell; b_1,\dots,b_n; - 2\zeta - \nu/2)$ is a solution
of the Bethe ansatz equations \eqref{tr.BAE} assigned to the triple
$(b_1,\dots,b_n; -2\zeta - \nu/2; V[-\nu])$.

\end{lem}

\begin{proof}
We have 
\bea
\Wr_x(p(x,U),q(x,U))= \frac {2\zeta + \ell-m}x\,
\Big(x^n + \sum_{s=1}^n\,(-1)^s\,a_s \,x^{n-s}\Big).
\eea
Now the lemma follows from Lemmas \ref{lem gen f}, \ref{lem WBA}.
\end{proof}

For $U\in \Om(\zeta, m,\ell)$ denote by $\mc G_U$ the monic differential operator
with kernel $U$, 
\bean
\label{GU}
\mc G_U = \der_x^2 + F_{1;U}(x)\der_x + F_{2,U}(x).
\eean
The operator $\mc G_U$ is 
obtained from the operator $\mc G$ by evaluating the generating functions $p,q$ 
at the point $U$.

\begin{lem}
\label{lem e=g}

Let $a\in X$ and $U\in \Wr^{-1}(a)$. Let $(t^0; b; 2\zeta + \nu/2)$ be the solution of the Bethe ansatz equations
described in Lemma \ref{cor gen f}. Let $\mc E_{t^0;\, b;\, 2\zeta + \nu/2}^c$ 
be the differential operator defined in \eqref{conj Dt}.
Then
\bea
\mc E_{t^0;\, b; \,2\zeta + \nu/2}^c = \mc G_U\,.
\eea
 \end{lem}
 
 \begin{proof}
 The lemma follows from Lemma \ref{kerDconj}.
  \end{proof}

\section{Isomorphisms}
\label{sec 11}

In Section \ref{sec 9} we introduced the $\B(\mu, m,\ell)$-module $\V^S[\nu]$, where
 $\mu\in \C$, \ $\nu = n-2m$, $m+\ell=n$.
In Section \ref{sec 10} we discussed the properties of
the algebra $\Oz$ under the assumption that 
$\zeta\in \C-\frac 12\Z$. 

\smallskip
 We  consider $\Oz$ as the $\Oz$-module with action defined by
 multiplication. 

\smallskip
In this section we construct an isomorphism between the $\B(\mu, m,\ell)$-module $\V^S[\nu]$
and the $\Oz$-module $\Oz$ under the assumption that
\bean
\label{ass}
\zeta = \frac\mu 2 - \frac \nu 4\quad \on{and} \quad \zeta\in \C-\frac 12\Z\,,
\eean
where the last inclusion can be reformulated as 
\bean
\label{assm}
\mu\notin\ \frac n2 + \Z\,,
\eean
cf. the
assumptions on $\mu$ and $\zeta$ in Theorems \ref{thm isom mu}, \ref{quasipthm}, 
 Lemmas \ref{lem WBA}, \ref{lem yty} and Section \ref{sec 10}.

\smallskip
The construction of the isomorphism is similar to the constructions in \cite{MTV3, MTV2}.

\subsection{Isomorphism of algebras}

Consider the map

\bea
\tau  :  \Oz \to \B(\mu, m,\ell), \quad G_{ij}\mapsto F_{ij}.
\eea
${}$

\begin{thm}
[cf. {\cite[Theorem 5.3]{MTV3}}, {\cite[Theorem 6.3]{MTV2}}]
\label{thm isoa}

Under the assumptions \eqref{ass}
the map $\tau$ is a well-defined 
isomorphism of graded algebras.

\end{thm}

\begin{proof} 
Let a polynomial $R(G_{ij})$ in generators $G_{ij}$  be equal to zero in
$\Oz$. Let us prove that the
corresponding polynomial $R(F_{ij})$ is equal to zero in $\Bm$.
Indeed, $R(F_{ij})$ is a polynomial
in $z_1,\dots,z_n$ with values in $\End(V[\nu])$. By Lemmas \ref{lem gen f} - \ref{lem e=g},
 \ref{basis}, for generic $b_1,\dots,b_n$ the
 value of the polynomial $R(F_{ij})$ at
$z_1 = b_1,\dots, z_n=b_n$ equals zero. Hence, the polynomial $R(F_{ij})$ equals zero identically and
the map $\tau$ is a well-defined defined homomorphism of algebras.

The elements $G_{ij}$, $F_{ij}$ are of the same degree. Hence  $\tau$ is a graded homomorphism.

Let a polynomial $R(G_{ij})$ in generators  $G_{ij}$ be a nonzero element of
$\Oz$. Then the value of
$R(G_{ij})$ at a generic point $U \in \Om(\zeta, m,\ell)$ is not equal to zero by Lemma \ref{lem e=g}.
Then the polynomial $R(F_{ij})$ is not identically equal to zero. Therefore, the map $\tau$ is injective.
Since the elements $F_{ij}$ generate the algebra $\Bm$, the map $\tau$ is surjective.
\end{proof}

The algebra $\C[z_1,\dots,z_n]^S$
 is embedded into the algebra $\Bm$ as the subalgebra of operators
of multiplication by symmetric polynomials.
 The
algebra $\C[z_1,\dots,z_n]^S$ is embedded into the algebra $\Oz$, the elementary symmetric polynomials
$\si_1(z),\dots,\si_n(z)$ being mapped to the elements
$\Si_1,\dots,\Si_n$. These
embeddings give the algebras $\Bm$ and $\Oz$ the structure of $\C[z_1,\dots,z_n]^S$-modules.

\begin{lem}
[{\cite[Lemma 6.4]{MTV3}}]
\label{lem iso m}
Under assumptions \eqref{ass} the map $\tau$ is an isomorphism of $\C[z_1,\dots,z_n]^S$-modules.

\end{lem}

\begin{proof}
The lemma follows from formulas \eqref{Wr.eqn2}, \eqref{G1}.
\end{proof}

\subsection{Isomorphism of modules}
\label{sec imo}

The subspace of $\V^S[\nu]$ of all elements of degree $0$ is of dimension one and is generated by the vector
\bea
v_+ = \sum_{I=(I_1,I_2),\, |I_1|=m, |I_2|=\ell} v_I\,.
\eea
The subspace of $\Oz$ of all elements of degree $0$ is of dimension one and is generated by the element $1$.
Define the $\C[z_1,\dots,z_n]^S$-linear map

\bean
\label{ups}
\phi : \Oz \to \V^S[\nu], \quad G\mapsto \tau(G)\,v_+\,.
\eean
\medskip

\begin{thm}
[{\cite[Theorem 6.7]{MTV3}}]

\label{thm ups}

Under assumptions \eqref{ass}, the map $\phi$ is a graded isomorphism of graded $\C[z_1,\dots,z_n]^S$-modules. 
The maps $\tau$ and $\phi$ intertwine the action of multiplication
operators on $\Oz$  and the action of the Bethe algebra $\Bm$  on $\V^S[\nu]$, that is, for any 
$f,g \in\Oz$, we have

\bean
\label{inter}
\phi(fg) = \tau(f)\,\phi(g).
\eean
${}$

\noindent
In other words, the maps $\tau$ and $\phi$ define an isomorphism between the $\Oz$-module $\Oz$ and the
$\Bm$-module $\V^S[\nu]$.

\end{thm}

\begin{proof}
First we show that the map $\phi$ is injective. Indeed,
the algebra $\Oz$ is a free polynomial algebra containing the subalgebra $\C[z_1,\dots,z_n]^S$.
The quotient algebra $\Oz/\C[z_1,\dots,z_n]^S$ is finite-dimensional by Lemma \ref{lem pdeg}.
The  kernel of $\phi$ is a proper ideal $\mc I$ in $\Oz$. Then $\tau(\mc I)$ is an ideal in $\Bm$.
Any proper ideal in $\Bm$ has zero intersection with $\C[z_1,\dots,z_n]^S$. Hence
$\mc I$ has zero intersection with $\C[z_1,\dots,z_n]^S$ and therefore is 
the zero ideal. The injectivity is proved.

The map $\phi$ is graded. 
The graded characters of $\V^S[\nu]$ and $\Oz$ are equal by Lemmas \ref{lem frV}    and  \ref{lem grO}.
Hence $\phi$ is an isomorphism.
\end{proof}

\begin{cor}
\label{lem isofi}
Assume that $a=(a_1,\dots,a_n)\in \C^n$ is such that the polynomial 
$x^n + \sum_{s=1}^n(-1)^s a_sx^{n-s}$ has distinct roots
$b_1,\dots,b_n$. Then under  assumptions \eqref{ass}, the isomorphisms 
$\tau$, $\phi$ induce the isomorphism of  the $\B(b_1,\dots,b_n;\mu; V[\nu])$-module
$V[\nu]$ and the $\O_a(\zeta;m,\ell)$-module $\O_a(\zeta;m,\ell)$, where
$\O_a(\zeta;m,\ell)$ is the algebra of functions on the fiber $\Wr^{-1}(a)$
of the Wronski map, see  \eqref{ide}.
\end{cor}

\begin{proof}
The corollary follows from Lemma \ref{lem fib} and Theorems \ref{thm isoa}, \ref{thm ups}.
\end{proof}

\begin{cor}
\label{cor deg}
The degree of the Wronski map $\Wr$ equals $\dim V[\nu] = \binom{n}{m}$.
\qed
\end{cor}

\subsection{Dynamical Bethe algebra and quasi-polynomials}

The space $V = V_1^{\ox n}$ has a nontrivial zero weight subspace if $n$ is even.
Let $n=2m$. For the zero weight subspace $V[0]$,
 we have $\nu=0, \,m=\ell$, and 
assumptions \eqref{ass} take the form
\bean
\label{asss}
\zeta = \frac\mu 2 \quad \on{and} \quad 
\mu\notin\  \Z\,.
\eean
\smallskip

Let $a=(a_1,\dots,a_n)\in \C^n$ be such that the polynomial 
$x^n + \sum_{s=1}^n(-1)^s a_sx^{n-s}$ has distinct nonzero roots
$b_1,\dots,b_n$. 
Consider the functional space $E[\mu]$ as the module
over the  dynamical Bethe algebra $\B(b_1,\dots,b_n;E[\mu])$, see
Section \ref{sec BE}.  Consider the
$\O_a(\zeta;m,m)$-module $\O_a(\zeta;m,m)$, where
$\O_a(\zeta;m,m)$ is the algebra of functions on the fiber $\Wr^{-1}(a)$
of the Wronski map.

\begin{cor}
\label{lem isofi}

Under  assumptions \eqref{asss}, the isomorphisms 
$\tau$, $\phi$ and the isomorphism $V[0]\to E[\mu]$  in Corollary \ref{iso B}
induce the isomorphism of the $\B(b_1,\dots,b_n;E[\mu])$-module $E[\mu]$ and
  the $\O_a(\zeta;m,m)$-module $\O_a(\zeta;m,m)$.
  \qed
  \end{cor}

\subsection{Weyl involution and transposition of quasi-polynomials}

Consider the 
\\
$\Bm$-module $\V^S[\nu]$ and  $\B(-\mu, \ell,m)$-module
$\V^S[-\nu]$. Consider the $\Oz$-module $\Oz$ and 
$\mc O(-\zeta, \ell, m)$-module $\mc O(-\zeta, \ell, m)$.

Under assumptions \eqref{ass}, consider the diagram,

\bean
\label{comd}
\begin{tikzcd}
 (\Bm,\,  \V^S[\nu])  \arrow[r, ] \arrow[d, ] & (\B(-\mu, \ell,m),\,\V^S[-\nu])  \arrow[d, ] 
\\
 (\Oz,\,  \Oz) \arrow[r,   ]  & (\mc O(-\zeta, \ell, m), \,\mc O(-\zeta, \ell, m)
)
\end{tikzcd}\ \   .
\eean
Here $\V^S[\nu] \to    \Oz$ and 
$\V^S[-\nu] \to \mc O(-\zeta, \ell, m)$ are  the module isomorphisms of Theorem \ref{thm ups}.
The map $ \V^S[\nu] \to \V^S[-\nu]$ is the module isomorphism of Lemma \ref{lem BB iso}.
The map    $\Oz \to  \mc O(-\zeta, \ell, m)$ is the module isomorphism defined by the transposition
of the quasi-polynomials $p,q$.

\begin{thm}
\label{thm tra}
The diagram \eqref{comd} is commutative.

\end{thm}

\begin{proof}
The theorem follows from Lemma \ref{e=e}.
\end{proof}

The commutativity of diagram \eqref{comd} implies the commutativity of the 
diagram of fibers over a generic point $a\in\C^n$,
\bean
\label{cofd}
\begin{tikzcd}
 (\B(b_1,\dots,b_n;\mu, V[\nu]), \,  V[\nu])  \arrow[r, ] \arrow[d, ] &
 (\B(b_1,\dots,b_n;-\mu, V[-\nu]), \,  V[-\nu])  \arrow[d, ] 
\\
   (\O_a(\zeta, m,\ell),\,\O_a(\zeta, m,\ell)) \arrow[r,   ]  & (\mc O_a(-\zeta, \ell, m),\,\O_a(-\zeta,\ell,m))
\end{tikzcd}  \ \ ,
\eean
see notations in Section \ref{sec imo}. 

Combining commutative diagrams \eqref{cofd} and \eqref{comD} we obtain the commutative
diagram
\bean
\label{cofd}
\begin{tikzcd}
 (\B(z; E[\mu]),\, E[\mu]) \arrow[r,   ]\arrow[d,]  & (\B(z; E[-\mu]),\, E[-\mu])\arrow[d, ]
\\
  (\O_a(\zeta, m,m),\,\O_a(\zeta, m,m)) \arrow[r,   ]  & (\mc O_a(-\zeta, m, m),\,\O_a(\zeta, m,m))
\end{tikzcd}  \ \ ,
\eean
which holds if $n=2m$ is even and $\mu\notin\Z$. The diagram 
identifies the Weyl involution $E[\mu]\to E[-\mu]$ in the functional spaces
of eigenfunctions of the KZB operator $H_0$ 
with the isomorphism $\O_a(\zeta, m,m)\to \O_a(-\zeta, m,m)$
induced by the transposition  of quasi-polynomials.

\end{document}